\documentclass[runningheads, envcountsame]{elsarticle}
\usepackage[draft]{microtype}

\usepackage{graphicx}
\usepackage{epic}
\usepackage{eepic}
\usepackage{url}
\usepackage{enumitem}
\usepackage{a4}

\usepackage{hyperref}
\usepackage{xspace}
\usepackage{amsmath}
\usepackage{amsthm}

\usepackage{marvosym}

\usepackage{amssymb}
\usepackage{stmaryrd}
\usepackage{tikz}
 \usetikzlibrary{trees}
 \usetikzlibrary{shapes}
 \usetikzlibrary{arrows}
 \usetikzlibrary{matrix}
\usetikzlibrary{fit}
\usetikzlibrary{backgrounds}	

 \usepackage{algorithm}
 \usepackage{algpseudocode}
 \usepackage{todonotes}

\newcommand{\threecolE}{\text{\sc 3-Colourability}^e}
\newcommand{\fourcolE}{\text{\sc 4-Colourability}^e}

\newcommand{\mbd}{\text{\sc Model-Based Diagnosis}^e}
\newcommand{\abd}{\text{\sc Abduction}^e}
\newcommand{\abddec}{\text{\sc Abduction}}

\newcommand{\repair}{\text{\sc Repair}^e}

\newcommand{\pExist}[1]{\text{\sc Exist}{\normalfont \_}#1}
\newcommand{\pCheck}[1]{\text{\sc Check}{\normalfont \_}#1}
\newcommand{\pExtendable}[1]{\text{\sc Extendable}{\normalfont \_}#1}
\newcommand{\pExtSol}[1]{\text{\sc ExtSol}{\normalfont \_}#1}

\newcommand{\pEnum}[1]{\text{\sc Enum}{\normalfont \_}#1}
\newcommand{\pCount}[1]{\text{\sc Count}{\normalfont \_}#1}
\newcommand{\pAnotherSol}[1]{\text{\sc AnotherSol}{\normalfont \_}#1}
\newcommand{\pExistAnotherSol}[1]{\text{\sc Exist-AnotherSol}{\normalfont \_}#1}
\newcommand{\pEnumSat}[1]{\mathrm{SAT}(#1)^e}

\newcommand{\sigmaSAT}[1]{\Sigma_{#1}\mathrm{SAT}}
\newcommand{\sigmaSATe}[1]{\Sigma_{#1}\mathrm{SAT}^e}
\newcommand{\piSATe}[1]{\Pi_{#1}\mathrm{SAT}^e}
\newcommand{\SAT}[1]{\mathrm{SAT}(#1)}
\newcommand{\SATe}{\mathrm{SAT}^e}
\newcommand{\circumscription}{\text{\sc Circumscription}}
\newcommand{\cardminsat}{\text{\sc CardMinSAT}}
\newcommand{\trans}{\mathrm{Trans}}

\newcommand{\domenum}{\text{\sc Dom-Enum}}
\newcommand{\transenum}{\text{\sc Trans-Enum}}

\newcommand{\calO}{{\mathcal O}}
\newcommand{\calA}{{\mathcal A}}
\newcommand{\calB}{{\mathcal B}}
\newcommand{\calC}{{\mathcal C}}

\newcommand{\calH}{{\mathcal H}}

\newcommand{\calP}{{\mathcal P}}

\newcommand{\SHP}{\ensuremath{\mathsf{\#P}}\xspace}
\newcommand{\SHdP}{\#\cdot\mathsf{P}}
\newcommand{\SHdNP}{\#\cdot\mathsf{NP}}
\newcommand{\SHdcoNP}{\#\cdot\mathsf{coNP}}
\newcommand{\SHdC}{\#\cdot{\cal C}}

\newcommand{\SHdSigmaPTwo}{\#\cdot\Sigma_{2}^{P}}
\newcommand{\SHdPiPTwo}{\#\cdot\Pi_{2}^{P}}

\newcommand{\ptime}{\text{\sf P}}

\newcommand{\Pcomp}{\text{\sc P}}
\newcommand{\NP}{\text{\sf NP}\xspace}
\newcommand{\EXP}{\text{\sf EXP}}
\newcommand{\EXPTIME}{\text{\sf EXPTIME}}

\newcommand{\coNP}{\text{\sf coNP}}
\newcommand{\PiP}[1]{\ensuremath{{\Pi}_{#1}^{P}}\xspace}
\newcommand{\SigmaP}[1]{\ensuremath{{\Sigma}_{#1}^{P}}\xspace}
\newcommand{\DeltaP}[1]{\ensuremath{{\Delta}_{#1}^{P}}\xspace}

\newcommand{\newoutput}{\mathsf{NOO}}
\newcommand{\ase}{\text{\sc AnotherSolExt}_{R}^{\mathcal{A}}}
\newcommand{\DTIME}{\text{\sf DTIME}}
\newcommand{\poly}{\operatorname{poly}}

\newcommand{\totalP}{\text{\sf TotalP}}

\newcommand{\outputP}{\text{\sf OutputP}}

\newcommand{\delayP}{\text{\sf DelayP}}

\newcommand{\incP}{\text{\sf IncP}}

\newcommand{\nop}[1]{}

\newcommand{\var}{\mathrm{var}}
\newcommand{\horn}{\protect\ensuremath{\mathrm{Horn}}}
\newcommand{\dualhorn}{\protect\ensuremath{\mathrm{dual Horn}}}

\newcommand{\bijunctive}{\protect\ensuremath{\mathrm{bijunctive}}}
\newcommand{\affine}{\protect\ensuremath{\mathrm{affine}}}

\newcommand{\T}{\protect\ensuremath{\mathrm{T}}}
\newcommand{\F}{\protect\ensuremath{\mathrm{F}}}
\newcommand{\Imp}{\protect\ensuremath{\mathrm{Imp}}}

\newtheorem{theorem}{Theorem}
\newtheorem{lemma}[theorem]{Lemma}
\newtheorem{proposition}[theorem]{Proposition}
\newtheorem{corollary}[theorem]{Corollary}

\newdefinition{example}{Example}
\newdefinition{definition}{Definition}
\begin{document}

\begin{frontmatter}

\title{A Complexity Theory for Hard Enumeration Problems}

\author[amu]{Nadia Creignou}
\ead{nadia.creignou@univ-amu.fr}

\author[tuw]{Markus Kr\"{o}ll}
\ead{kroell@dbai.tuwien.ac.at}

\author[tuw]{Reinhard Pichler}
\ead{pichler@dbai.tuwien.ac.at}

\author[tuw]{Sebastian Skritek}
\ead{skritek@dbai.tuwien.ac.at}

\author[luh]{Heribert Vollmer}
\ead{vollmer@thi.uni-hannover.de}

\address[tuw]{TU Wien, Vienna, Austria}
\address[amu]{Aix-Marseille Univ, CNRS, Marseille, France}
\address[luh]{Leibniz Universit\"at Hannover, Hannover, Germany}

\begin{abstract}
Complexity theory provides a wealth of complexity classes for 
analyzing the complexity of decision and counting problems. 
Despite the practical relevance of enumeration problems, the tools
provided by complexity theory for this important class of problems
are very limited.
In particular, complexity classes analogous to the polynomial hierarchy
and an appropriate notion of problem reduction are missing. 
In this work, we lay the foundations for a complexity theory of 
hard enumeration problems by proposing a hierarchy of complexity classes and 
by investigating notions of reductions for enumeration problems. 
\end{abstract}
\end{frontmatter}

\section{Introduction}
\label{sect:introduction}

While decision problems often ask for the {\em existence of a solution\/} to some problem instance, enumeration problems aim at outputting {\em all solutions\/}. 
In many domains, enumeration problems are thus the most natural kind of problems.
Just take the database area (usually the user is interested in all answer tuples and 
not just in a yes/no answer) or diagnosis (where the user wants to retrieve possible
explanations, and not only whether one exists) as two examples.
Nevertheless, the complexity of enumeration problems is far less studied than the complexity of decision problems.

It should be noted that even simple enumeration problems may produce big output. 
To capture the intuition of easy to enumerate problems 
-- despite a possibly exponential number of output values --
various notions of tractable enumeration classes have been proposed in \cite{JPY1988}.
The class $\delayP$ (``polynomial delay'')
contains all enumeration problems where, for given instance $x$, 
(1) the time to compute the first solution, (2) the time between outputting any two consecutive solutions, and (3) the time to detect that no further solution exists, are all polynomially bounded in the size of $x$. 
The class $\incP$ (``incremental polynomial time'')
contains those enumeration problems where, for given instance $x$, 
the time to compute the next solution and for detecting that no further solution exists
is polynomially bounded in the size of both $x$ and of the already computed solutions. 
Obviously, the relationship $\delayP \subseteq \incP$
 holds.
In \cite{Strozecki2010}, the proper inclusion $\delayP \subsetneq \incP$
is mentioned. 
For these tractable enumeration classes, a variety of membership results exist, a few
examples are given in 
\cite{DBLP:journals/jcss/LucchesiO78,DBLP:journals/tods/KimelfeldK14,DBLP:journals/fuin/CreignouV15,DBLP:conf/pods/CarmeliKK17,bagan2007acyclic,DSS2014}

There has also been work on intractable enumeration problems. 
Intractability of enumeration is typically proved by showing intractability of a 
related decision problem rather than directly proving lower bounds by relating one enumeration problem to the other. 
Tools for a more fine-grained analysis of intractable enumeration problems are 
missing to date. For instance, up to now we are not able to make a differentiated analysis of the complexity of the following typical enumeration problems:

\medskip
\fbox{
\begin{tabular}{ll}
\multicolumn{2}{l}
{$\piSATe{k}$ / $\sigmaSATe{k}$} \\
  {Instance}: & $\psi=\forall y_1\exists y_2\ldots Q_k y_k\phi(\vec{x},\vec{y})$ 
  /  $\psi=\exists y_1\forall y_2\ldots Q_k y_k\phi(\vec{x},\vec{y})$\\
  {Output}: & All assignments for $\vec{x}$ such that $\psi$ is true.\\
\end{tabular}
}
\medskip

This is in sharp contrast to decision problems, where the 
polynomial hierarchy is crucial for a detailed 
complexity analysis. As a matter of fact, it makes a big difference,
if an \NP-hard problem is in \NP or not.
Indeed, \NP-complete problems have an efficient transformation into SAT and can therefore be solved by making use of powerful SAT-solvers. 
Similarly, problems in \SigmaP{2} can be solved by using ASP-solvers. 
Finally, also for problems on higher levels of the polynomial hierarchy, the 
number of quantifier alternations in the QBF-encoding matters when using 
QBF-solvers. 
For counting problems, an analogue of the 
polynomial hierarchy has been defined 
in form of the 
$\SHdC$--classes with 
${\cal C} \in \{\ptime, \coNP, \PiP{2}, \dots\}$
\cite{DBLP:journals/sigact/HemaspaandraV95,DBLP:journals/siamcomp/Toda91}.
For enumeration problems, no such analogue has been studied.

\smallskip
\noindent
{\bf Goal and Results.} 
The goal of this work is to lay the foundations for a complexity theory of hard enumeration problems by defining  appropriate complexity classes for intractable enumeration
and a suitable notion of problem reductions.
We propose to extend tractable enumeration classes by
oracles. We will thus get a hierarchy of classes $\delayP^\calC$, $\incP^\calC$, where
various complexity classes ${\cal C}$ are used as oracles. 
As far as the definition of an appropriate notion of reductions is concerned, 
we follow the usual philosophy of reductions: if some enumeration problem can
be reduced to another one, then we can use this reduction together with an 
enumeration algorithm for the latter problem to solve the first one.
We observe that two principal kinds of reductions are used for decision problems, 
namely many-one reductions and Turing reductions.
Similarly, we shall define a more declarative-style and a more procedural-style
notion of reduction for enumeration problems. 
Our results are summarized below.
%  All missing proof details can be found in the
% full version of this article \cite{fullversion}.

\begin{itemize}%[leftmargin=*,label=$\bullet$]
\item 
{\em Enumeration complexity classes.} 
In  Section \ref{sect:classes}, we introduce a hierarchy of complexity classes of intractable enumeration via oracles and prove that it is strict unless the polynomial hierarchy collapses.

\item
{\em Declarative-style reductions.} 
In Section \ref{sect:first-reduction}, 
we introduce a declarative-style notion of reductions. While  they
enjoy some desirable properties, we   do not succeed in exhibiting complete problems under this type of reductions.

\item
{\em Procedural-style reductions.} 
In Section \ref{sect:second-reduction}, we introduce a procedural-style notion of reductions,
show that they have the desirable properties any reduction should possess,
and show that the enumeration problems associated with the typical quantified Boolean satisfiability problems are complete in our hierarchy of enumeration problems.

\item
{\em Completeness results for natural problems.}
Starting with these completeness results, we establish a chain of reductions among several natural enumeration problems from areas such as generalized satisfiability, circumscription, model-based diagnosis, abduction, and repairs of inconsistent databases in Section~\ref{sect:completeness}, thus proving completeness of these problems in classes of our hierarchy. 

\end{itemize}

This work is an extension of the conference article \cite{DBLP:conf/lata/CreignouKPSV17}.

\section{Preliminaries}
\label{sect:preliminaries}

In the following, $\Sigma$ denotes a finite alphabet and $R$ denotes a polynomially bounded, binary relation
$R\subseteq \Sigma^*\times\Sigma^*$, i.e., there is a polynomial $p$
such that for all $(x,y)\in R$, $|y|\leq p(|x|)$ holds.
For every string $x$, $R(x)=\{y\in \Sigma^*\mid (x,y)\in R\}$.
A string $y\in R(x)$ is called a \emph{solution} for $x$.
With 
% a polynomially bounded, binary 
such a relation $R$, we can associate several natural computational problems:
\medskip

\framebox[10cm][l]{\begin{tabular}{ll}
 \multicolumn{2}{l}{$\pExist{R}$} \\
  {Instance}: &  $x\in\Sigma^*$ \\ 
  {Question}: &  Is $R(x)\neq\emptyset$?
\end{tabular}}

\framebox[10cm][l]{\begin{tabular}{ll}
 \multicolumn{2}{l}{$\pExistAnotherSol{R}$} \\
  {Instance}: & $x \in \Sigma^*, Y\subseteq R(x)$ \\ 
  {Output}: & Is $(R(x)\setminus Y) \neq \emptyset$?
 \end{tabular}}

\framebox[10cm][l]{\begin{tabular}{ll}
 \multicolumn{2}{l}{ $\pAnotherSol{R}$} \\
  {Instance}: & $x \in \Sigma^*, Y\subseteq R(x)$ \\ 
  {Output}: &  $y\in R(x)\!\setminus\!Y$ 
  or declare that no such $y$ exists.
 \end{tabular}}

\framebox[10cm][l]{\begin{tabular}{ll}
 \multicolumn{2}{l}{$\pCheck{R}$}\\
  {Instance}: & $(x,y) \in \Sigma^*\times\Sigma^*$\hspace{-.1em} \\
  {Question}: & Is $(x,y)\in R$?\\
\end{tabular}}

\framebox[10cm][l]{\begin{tabular}{ll}
 \multicolumn{2}{l}{$\pExtSol{R}$} \\
  {Instance}: &  $(x,y) \in \Sigma^*\times\Sigma^*$ \\ 
  {Question}: &  Is there some $y'\in\Sigma^*$ such that $(x,yy')\in R$? 
\end{tabular}}
\bigskip

\noindent
In counting complexity, the following computational problem arises:

\smallskip
\framebox[10cm][l]{\begin{tabular}{ll}
 \multicolumn{2}{l}{$\pCount{R}$} \\
  {Instance}: &  $x\in\Sigma^*$ \\ 
  {Output}: &  $\bigl|\{ y\in\Sigma^* \mid (x,y)\in R\}\bigr|$.\\
\end{tabular}}
\bigskip

\noindent
A binary relation $R$ also gives rise to an enumeration problem, which 
aims at outputting the 
set $R(x)$ of solutions for $x$.
%function 
%$\mathsf{Sol}_R: \Sigma^*\rightarrow 2^{\Sigma^*}, x\mapsto \{y\in\Sigma^*\mid (x,y)\in R\}$.
%
%\todo[inline]{NC: we have a double notation, $R(x)$ and $\mathsf{Sol}_R(x)$, I would be in favor to skip the second one in all the text. }
%\medskip

\smallskip
\framebox[10cm][l]{\begin{tabular}{ll}
 \multicolumn{2}{l}{$\pEnum{R}$} \\
  {Instance}: &  $x\in\Sigma^*$ \\ 
  {Output}: &  $R(x)=\{ y\in\Sigma^* \mid (x,y)\in R\}$.\\
\end{tabular}}
\bigskip

We assume the reader to be familiar with the polynomial hierarchy --
the complexity classes $\ptime$, $\NP$, $\coNP$ and, more generally, 
$\DeltaP{k}$, $\SigmaP{k}$, and $\PiP{k}$ for $k \in \{0,1, \dots\}$.

%\todo[inline]{NC: should we develop? }

As we usually deal with exponential runtime of algorithms in the course of an enumeration process,
we will also make use of the weak $\EXP$ hierarchy (see \cite{hemachandra1987strong}) in the following. It is an exponential time analogue
of the polynomial hierarchy defined as follows:

\begin{definition}[Exponential Hierarchy]
\label{def:exp-hierarchy}
\begin{align*}
\Delta_0^\EXP &=\Sigma_0^\EXP = \Pi_0^\EXP = \EXP\\
\Delta_k^\EXP &= \EXP^{\SigmaP{k-1}},\text{ for }k\geq1\\
\Sigma_k^\EXP &= \text{\sf NEXP}^{\SigmaP{k-1}},\text{ for }k\geq1\\
\Pi_k^\EXP &= \text{\sf coNEXP}^{\SigmaP{k-1}},\text{ for }k\geq1
\end{align*}
\end{definition}

\smallskip

A hierarchy of counting
complexity classes similar to the polynomial hierarchy was defined 
\cite{DBLP:journals/sigact/HemaspaandraV95,DBLP:journals/siamcomp/Toda91}. 
Let ${\cal C}$ be a complexity class of decision problems. Then 
$\SHdC$ denotes the class of all counting problems ${\cal P}$, such that 
there exists a relation $R$ with $\pCheck{R} \in {\cal C}$ and ${\cal P} = \pCount{R}$.
For $\SHP$ -- the best known counting complexity class --  
we have $\SHP = \SHdP$.
Moreover, the inclusions  
$\SHP \subseteq \SHdNP \subseteq \SHdcoNP \subseteq \SHdSigmaPTwo 
\subseteq \SHdPiPTwo \subseteq \dots $
hold.

\smallskip

A complexity class $\calC$ is \emph{closed under a reduction $\leq_r$}
if, for any two binary relations $R_1$ and $R_2$ we have that
$R_2 \in \calC$ and $R_1 \leq_r R_2$ imply $R_1 \in \calC$. Furthermore,
a reduction $\leq_r$ is \emph{transitive} if for any three binary relations
$R_1, R_2, R_3$, it is the case that $R_1 \leq_r R_2$ and $R_2 \leq_r R_3$ implies
$R_1 \leq_r R_3$.

\section{Complexity Classes for Enumeration}
\label{sect:classes}

In Section~\ref{sect:introduction}, 
we have already recalled two important tractable enumeration complexity classes,
$\delayP$ and $\incP$
from \cite{JPY1988}. 
Note that in \cite{Strozecki2010,DBLP:journals/mst/Strozecki13}, 
these classes are defined slightly differently by restricting 
$\delayP$ and $\incP$ to those problems $\pEnum{R}$ such that 
the corresponding problem $\pCheck{R}$ is in $\ptime$ --- in addition
to the constraints on the allowed delays.
% by allowing only
% those $\pEnum{R}$ problems in $\delayP$ and $\incP$
%
% where  the corresponding $\pCheck{R}$ problem is in $\ptime$. 
We adhere to the definition of tractable enumeration classes 
from \cite{JPY1988}. 

%\todo[inline]{Should we give here a formal definition of these two classes}

In contrast to counting complexity, defining a hierarchy of enumeration  problems via the $\pCheck{R}$ problem of binary relations $R$ does not seem appropriate.
%This can be seen by considering  artificial problems obtained by padding the set of 
%solutions of any problem with  an exponential number of fake (and trivial to produce) 
%solutions. While these fake solutions do not change the complexity of the check problem,
%enumerating these exponentially many fake solutions first gives an enumeration algorithm
%enough time to search for the nontrivial ones.
%
Note that while the same counting problem $\calP$ can be defined by different relations, i.e.,
${\cal P} = \pCount{R_1}$ and ${\cal P} = \pCount{R_2}$ for $R_1\neq R_2$, the relation associated with an enumeration problem is fixed. 

Thus, we need an alternative approach for defining meaningful enumeration complexity classes. 
To this end, we first fix our computation model:
Observe that there is a subtle difference between how Turing Machines and
random access machines (RAMs) can access data in their ``memory''. Due to the linear
nature of the tapes of Turing Machines, accessing an exponential size data structure
requires exponential time, even if just a small portion of its data is actually read.
This is usually not a problem when
studying decision or search problems, since writing an exponentially sized data structure
onto the tape already requires exponential time, and as a result the additional exponential
time to read the data structure has no more effect on the overall runtime. 

This situation, however, changes for enumeration problems and polynomial delay, since it 
is now possible that, while computing an exponential number of solutions, an exponential
size data structure accumulates. Thus computing another solution with only polynomial
delay is not possible with Turing Machines if this requires some information from such a
data structure. However, by maintaining suitable index structures, it might be possible to
gather the necessary information in polynomial time from such a data structure on a 
random access machine (RAM) that allows to access its memory directly and not first needs
to move the read/write head over the correct tape position.
In fact, due to this property that allows to access (polynomially sized) parts of exponential
size data in polynomial time, it is common to use the RAM model as a computational model
for the study of enumeration problems (cf.\ \cite{Strozecki2010}).
% , 
% because a RAM can access parts of exponential-size data in polynomial time.
We restrict ourselves here to polynomially bounded RAMs, i.e., 
throughout the computation of such a machine, the size of the content of each register is polynomially bounded in the size of the input.

For enumeration, we will also make use of RAMs with an 
$\mathsf{output}$-instruction, as defined in \cite{Strozecki2010}. This model can be extended 
further by introducing decision oracles. The input to the oracle is stored in 
special registers and the oracle takes consecutive non-empty registers as input. 
Moreover, following \cite{book/AroraBarak2009},
we use a computational model
that does not delete the input of an oracle call once such a call is made. For a detailed
definition, refer to \cite{Strozecki2010}.
It is important to note that due to the exponential runtime of an enumeration algorithm and the 
fact that the input to an oracle is not deleted when the oracle is executed, 
the input to an oracle call may eventually become exponential as well. 
Clearly, this can only happen if exponentially many consecutive special registers are non-empty, since we assume
also each special register to be polynomially bounded.

Using this we
define a collection of enumeration complexity classes via oracles: 

\begin{definition}[enumeration complexity classes]
\label{def:enumeration-classes-with-oracles}
Let $R$ be a binary relation, and $\calC$ a decision complexity class. 
Then we say that:
\begin{itemize}
 \item $\pEnum{R}\in\delayP^\calC$ if there is a RAM $M$ with an oracle $L$ in $\calC$ such that
 for any instance $x$,  $M$ enumerates $R(x)$
 with polynomial delay. The class $\incP^\calC$ 
is defined analogously.
\item $\pEnum{R}\in\delayP_p^\calC$ if
there is a RAM $M$ with an oracle $L$ in $\calC$
such that for any instance $x$,  $M$ enumerates $R(x)$ with polynomial delay and 
the size of the input to every oracle call is polynomially bounded in $|x|$.
\end{itemize}
\end{definition}

\noindent
Note that the restriction of the oracle inputs to polynomial size 
makes a crucial difference when it comes to $\delayP^\calC$,
where we have a discrepancy between the
polynomial restriction (w.r.t. the input $x$) 
on the time between two consecutive solutions are output
and the possibly exponential size (w.r.t. the input $x$) of oracle calls. 
No such discrepancy exists for 
$\incP^\calC$,
where the same polynomial upper bound w.r.t.\
the already computed solutions (resp.\ all solutions) applies both to the allowed
time and to the size of the oracle calls.
In fact, the lower computational power of $\delayP$ compared with $\incP$ can be compensated
by  equipping the lower class with a slightly more powerful oracle.

\begin{theorem}\label{theorem:deltacollapse}
Let $k\geq 0$. Then $\delayP^{\DeltaP{k+1}}= \incP^{\SigmaP{k}}$.
\end{theorem}
%\begin{proof}
%The inclusion $\delayP^{\DeltaP{k+1}}\subseteq\incP^{\SigmaP{k}}$ holds since the incremental delay
%with access to a $\SigmaP{k}$-oracle gives enough time to compute the answers of a $\DeltaP{k+1}$-oracle.
%To show that $\delayP^{\DeltaP{k+1}}\supseteq\incP^{\SigmaP{k}}$, let $\pEnum{R}\in\incP^{\SigmaP{k}}$ and
%$\calA$ be a corresponding enumeration algorithm. We define a decision problem 
%$\ase$ that, on an input $y_1,\ldots,y_n,y',x\in\Sigma^*$, decides whether $y'$ is the prefix of the
%$(n+1)$-st output of $\calA(x)$. Since $\calA$ witnesses the membership $\pEnum{R}\in\incP^{\SigmaP{k}}$,
%it follows that $\ase\in\DeltaP{k+1}$, and using this language as an oracle, we have that $\pEnum{R}\in\delayP^{\DeltaP{k+1}}$.
%\qed\end{proof}
\begin{proof}
Let $k\geq 0$.
We start the proof by showing that $\delayP^{\DeltaP{k+1}}\supseteq \incP^{\SigmaP{k}}$. 
So let $\pEnum{R}\in\incP^{\SigmaP{k}}$ with corresponding binary relation $R$
and let $x\in\Sigma^*$. Fix
an incremental delay algorithm $\calA$ which uses a $\SigmaP{k}$-oracle
witnessing the membership $\pEnum{R}\in\incP^{\SigmaP{k}}$ and let $<^*$ be
an order on $R(x)$ induced by algorithm $\calA$, i.e. the $i$-th output of $\calA$ on input $x$ is the
$i$-th element in $<^*$. We define the following decision problem:
\medskip 

\framebox[10cm][l]{\begin{tabular}{ll}
 \multicolumn{2}{l}{$\ase$} \\
  {Instance}: &  $y_1,\ldots,y_n,y',x\in\Sigma^*$ \\ 
  {Question}: & Is $y'$ a prefix of $y_{n+1}$, where $y_{n+1}$ is the\\
  &$(n+1)$-th element in $R(x)$ w.r.t. $<^*$?\\
\end{tabular}}

\medskip 

We first note that $\ase\in\DeltaP{k+1}$. Indeed, assume that we have given an instance 
$y_1,\ldots,y_n,y',x\in\Sigma^*$. Then we can use $\calA$ to enumerate the
first $n+1$ elements of $R(x)$ in time $\calO(\poly(|x|,|n+1|))=
\calO(\poly(|y_1|+\ldots+|y_n|+|y'|+|x|))$ and then check whether $y'$ is a prefix of $y_{n+1}$. As
$\calA$ uses a $\SigmaP{k}$-oracle, this decision can be made 
within $\ptime^{\SigmaP{k}}=\DeltaP{k+1}$. The membership 
$\pEnum{R}\in\delayP^{\DeltaP{k+1}}$ follows immediately, as we can construct a polynomial
delay algorithm with an $\ase$-oracle that enumerates $\pEnum{R}$ in a similar way to
commonly used enumeration
algorithms, see~\cite{Strozecki2010,DBLP:journals/fuin/CreignouV15} or Proposition~\ref{prop:new}.
This enumeration algorithm starts by computing some $a\in\Sigma$ such that the oracle on this input
returns 'yes'. Then $y_1$ (and also every other $y_i$) can be computed by
repeatedly extending the previous input to the oracle by the unique element in $\Sigma$ such that
the decision oracle $\ase$ returns 'yes'.\\
Next we need to show that $\delayP^{\DeltaP{k+1}}\subseteq \incP^{\SigmaP{k}}$,
so let $\pEnum{R}\in\delayP^{\DeltaP{k+1}}$, and let $\calA$ be an algorithm witnessing
this membership. Moreover, let $L\in\ptime^{\SigmaP{k}}$ be the language used for the 
$\DeltaP{k+1}$-oracle in $\calA$, with a polynomial $q$ and a language $L'\in\SigmaP{k}$ such
that $L\in\DTIME^{L'}(q(n))$. Let $p$ be be the polynomial for the delay of $\calA$.
We describe an algorithm $\calB$ that has an incremental delay of $p(X)\cdot Y\cdot q(p(X)\cdot Y)$
(here the indeterminate $X$ stands for the size of the input and the indeterminate $Y$
for the number of previously output solutions)
that uses a $\SigmaP{k}$-oracle, such that $\calB$ enumerates $\pEnum{R}$. For this,
let $x\in\Sigma^*$ and assume that we want to enumerate $R(x)$. The
algorithm $\calB$ works as follows:
\begin{itemize}
\item Let $y_1$ be the first element output by algorithm $\calA$ on input $x$. 
As this output can be computed by $\calA$ in time $p(|x|)$, at most $p(|x|)$ many calls 
to the $L$-oracle have been made with an input of size at most $p(|x|)$. Thus the
answer of every oracle call can be computed in time $q(p(|x|))$ using an $L'$-oracle.
Therefore $y_1$ can be output by $\calB$ in time $p(|x|)\cdot q(p(|x|))$ 
by running $\calA$ until the first output, and simulating
the oracle calls accordingly.
\item For $n\geq 2$, let $y_n$ be the $n$-th element output by $\calA$. As with $y_1$, we can
make the $n$-th output of $\calB$ by running $\calA$ until $y_n$ is output and simulating
the oracle calls accordingly. Indeed, $\calA$ takes $p(|x|)\cdot n$ steps to output $y_n$, 
with at most $p(|x|)\cdot n$ oracle calls to $L$  with an input of size bounded by $p(|x|)\cdot n$.
Thus $y_n$ can be computed in time $p(|x|)\cdot n\cdot q(p(|x|\cdot n))$ using an $L'$-oracle.
\end{itemize}
\end{proof}

In fact, given that the classes $\DeltaP{k}$ describe problems that can be solved in polynomial
time by a deterministic machine with a $\SigmaP{k-1}$-oracle, it is not very surprising that
for enumeration algorithms, $\DeltaP{k}$-oracles can be simulated by $\SigmaP{k-1}$-oracles,
as long as one accounts for the potentially accumulating input to the $\DeltaP{k}$-oracles.

This property is not restricted to the case shown in
Theorem~\ref{theorem:deltacollapse}, but applies to all $\DeltaP{k}$-oracles.
As a result, in the remainder of this work, we omit any result regarding the classes
$\delayP_p^{\DeltaP{k}}, \delayP^{\DeltaP{k}}$ or $\incP^{\DeltaP{k}}$,
as they are given implicitly by a result on
$\delayP_p^{\SigmaP{k}}$ or $\incP^{\SigmaP{k}}$ respectively.
Indeed, Theorem~\ref{theorem:deltacollapse} shows the equality of $\delayP^{\DeltaP{k}}$
and $\incP^{\SigmaP{k-1}}$.
Concerning the remaining classes, observe that 
$\delayP_p^{\ptime^{\calC}} = \delayP_p^{\calC}$ for any decision complexity class $\calC$; 
in particular we obtain
$\delayP_p^{\DeltaP{k}} = \delayP_p^{\SigmaP{k-1}}$.
A similar results holds for the $\incP^{\calC}$ classes even without the bound on the input 
size to the decision oracle. This is due to the fact that the size of any input $\alpha$ to the 
decision oracle is at most polynomial in the combined size of the instance and number of 
outputs. Thus an incremental delay gives us enough time to make 
$\mathsf{poly}(|\alpha|)$ many oracles calls, meaning that 
$\incP^{\DeltaP{k}}=\incP^{\SigmaP{k-1}}$. 

\subsection{A Hierarchy of Enumeration Complexity Classes}

We now prove that our classes provide strict hierarchies under the assumption that the polynomial hierarchy is strict.

\begin{theorem}\label{theorem:delaypccollapse}
Let $k\geq0$.  Then, unless the polynomial hierarchy collapses to the 
$(k+1)$-st
level, 
\begin{align*}
&\delayP_p^{\SigmaP{k}}\subsetneq\delayP_p^{\SigmaP{k+1}},
\delayP^{\SigmaP{k}}\subsetneq\delayP^{\SigmaP{k+1}},
\incP^{\SigmaP{k}}\subsetneq\delayP^{\SigmaP{k+1}}\\
&\text{ and }
\incP^{\SigmaP{k}}\subsetneq\incP^{\SigmaP{k+1}}.
\end{align*}

\end{theorem}
\begin{proof}

Let $k\geq0$.
 By Theorem~\ref{theorem:deltacollapse} and the definition of our classes we have that
\begin{equation*}
 \delayP^{\SigmaP{k}}\subseteq\delayP^{\DeltaP{k+1}}= \incP^{\SigmaP{k}}\subseteq \delayP^{\SigmaP{k+1}}
 \subseteq  \incP^{\SigmaP{k+1}},
 \end{equation*} 
so we only need to show that 
$\delayP_p^{\SigmaP{k}}\subsetneq\delayP_p^{\SigmaP{k+1}}$ and  $\incP^{\SigmaP{k}}\subsetneq\delayP^{\SigmaP{k+1}}$.
 Let $L$ be a $\SigmaP{k+1}$-complete problem. Define a relation $R_L=\{(x,1)\mid x\in L\}$. It
is clear that $\pCheck{R_L}$ is $\SigmaP{k+1}$-complete. Moreover, the enumeration
problem $\pEnum{R_L}$ is in $\delayP_p^{\SigmaP{k+1}}$ (thus also in $\delayP^{\SigmaP{k+1}}$).
Assume that $\pEnum{R_L}\in\delayP_p^{\SigmaP{k}}$ (or $\pEnum{R_L}\in\incP^{\SigmaP{k}}$).
Then, as there is only one output to the enumeration problem, $\pCheck{R_L}$ can be decided in polynomial time using a $\SigmaP{k}$-oracle,
meaning that $\pCheck{R_L}\in\DeltaP{k+1}$ and thus the polynomial hierarchy  collapses to the
$(k+1)$-st level.
\end{proof}

The statement of the following proposition is twofold: First, it shows that the complexity classes based on 
$\delayP$ with a polynomial bound on the decision oracle and the complexity classes based on $\delayP$
without such a bound, respectively, are very likely to be distinct. Second,
the gap between $\delayP_p^{\SigmaP{k}}$
and $\delayP^{\SigmaP{k}}$ cannot be overcome by equipping the oracle-bounded enumeration complexity class with
a slightly more powerful oracle, in contrast to the result of Theorem~\ref{theorem:deltacollapse}.
Note that the proposition refers to the weak $\EXP$ hierarchy as defined in Definition \ref{def:exp-hierarchy}.

\begin{proposition}\label{prop:delayppsubdelayp}
Let $k\geq 1$. If $\EXP\subsetneq\Delta_{k+1}^\EXP$, then 
\begin{equation*}
\delayP_p^{\SigmaP{k}}\subsetneq\delayP^{\SigmaP{k}}\not\subseteq\delayP_p^{\SigmaP{k+1}}.
\end{equation*}
\end{proposition}

To show Proposition~\ref{prop:delayppsubdelayp}, we first need to prove the following result.

\begin{lemma}\label{lemma:inccheckexp}
Let $R$ be a binary relation and $k\geq 0$. If $\pEnum{R}\in \delayP_p^{\SigmaP{k}}$, 
then  $\pCheck{R}\in\EXP$.
\end{lemma}

\begin{proof}
Let $(x,y)$ be an instance of $\pCheck{R}$ and let $k\geq0$. 
Further let $\calA$ be an enumeration algorithm
witnessing the membership $\pEnum{R}\in \delayP_p^{\SigmaP{k}}$. In order to decide whether %$y\inR(x)$
$(x,y)\in R$, we simply
enumerate all of %$R(x)$.
$R(x)$ and check whether $y\in R(x)$.
Let $q,r$ be a polynomials such
that the decision of the $\SigmaP{k}$-oracle can be computed in $\calO(2^{q(n)})$,
and the size of the input to any such decision oracle while executing $\calA$ is bounded by $r$.
By the definition
of polynomial delay and the fact that $R$ is a polynomial relation, there is some polynomial $h$
such that %$R(x)$ 
$R(x)$ can be enumerated in time $\calO(2^{q(r(n))+h(n)})$, i.\,e., in exponential time.
\end{proof}

\begin{proof}[of Proposition~\ref{prop:delayppsubdelayp}]

 Assume that $\EXP\subsetneq\Delta_{k+1}^\EXP$. Then there exists some polynomial $q$ and a 
 language $L$ such that $L\in \Delta_{k+1}^\EXP\setminus\EXP$ and $L$ can be decided in 
 time $\calO(2^{q(n)})$ using a $\SigmaP{k}$-oracle. Define the following enumeration problem:
 
 \medskip 
 
\framebox[10cm][l]{\begin{tabular}{ll}
 \multicolumn{2}{l}{$\pEnum{D_0}$} \\
  {Instance}: & $x\in\Sigma^*$.\\
  {Output}: & All $\{0,1\}$-words of length $q(|x|)$, and $2$ if $x\in L$\\
\end{tabular}}

\medskip
 
 First note that $\pEnum{D_0}\in  \delayP^{\SigmaP{k}}$ by an algorithm $\calA$ that enumerates
 all $2^{q(|x|)}$ words in $\{0,1\}^{q(|x|)}$ in $\calO(2^{q(|x|)})$. While enumerating the trivial
 part of the output, $\calA$ also has enough time to compute whether $x\in L$, and then
 makes the last output ('2' or nothing) accordingly. 
Next assume that $\pEnum{D_0}\in  \delayP_p^{\SigmaP{k}}$ (or $\delayP_p^{\SigmaP{k+1}}$). Then,
by Lemma~\ref{lemma:inccheckexp}, $\pCheck{D_0}\in\EXP$. Therefore we can check for all $x\in\Sigma^*$ whether
$(x,2)\in D_0$, which is equivalent to $x\in L$. Thus we can decide $L$ in exponential time,
a contradiction. This proves also the second claim.
\end{proof}

 By generalizing the statement of Lemma~\ref{lemma:inccheckexp}, we even have that
$\delayP^{\calC}_p$ and $\delayP^{\SigmaP{k}}$ are incomparable for any $k\geq 1$ and any $\calC$ 
within $\EXP$ under the assumption that the weak $\EXP$ hierarchy does not collapse.
This is due to the fact that $\pCheck{R}$ for an enumeration problem
$\pEnum{R}\in\delayP^{\calC}_p$ is within $\EXP$, whereas $\delayP^{\SigmaP{k}}$ contains an
enumeration problem with corresponding check problem in $\Delta_{k+1}^{\EXP}\setminus\EXP$.\\

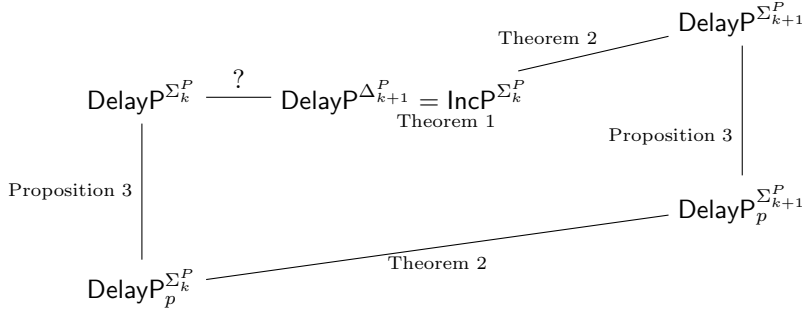
\begin{figure}[t]
\begin{center}
\begin{tikzpicture}
  \matrix (m) [matrix of nodes, row sep=1em,
    column sep=2.5em]{
    &  & & \node (1-4) {$\delayP^{\SigmaP{k+1}}$};\\
    \node (2-1) {$\delayP^{\SigmaP{k}}$}; &  \node (2-3) {$\delayP^{\DeltaP{k+1}}=\incP^{\SigmaP{k}}$}; &\\
    &  &  &\\
%    & & &\\
    & 
    % \node (3-2) {$\incP_p^{\SigmaP{k}}$};
    & & \node (3-4) {$\delayP_p^{\SigmaP{k+1}}$};\\
    \node (4-1) {$\delayP_p^{\SigmaP{k}}$}; & & \\};
    
  % \node[text=blue!50] at (3,-3) {Oracle bounded};
    
  \path            
  	(2-3) edge node[above] {?} (2-1) 
   				  edge node[left, yshift = 0.6em, xshift = 0.5em] {\scriptsize Theorem~\ref{theorem:delaypccollapse}} (1-4)
				  % edge node[right] {} (3-2)
      
	% (4-1) edge node[left] {} (2-1)

    % (4-1) edge node[above] {\textcolor{red}{?}}  (3-2)
    (4-1)	 edge node[left] {\scriptsize Proposition~\ref{prop:delayppsubdelayp}} (2-1)
    		 edge node[below] {\scriptsize Theorem~\ref{theorem:delaypccollapse}} (3-4)

	(1-4) edge node[left,yshift=-1em] (h1) {\scriptsize Proposition~\ref{prop:delayppsubdelayp}} (3-4);
    
    % (3-2) edge [densely dotted] node[above] {\textcolor{red}{?}}  (3-4);    
    
    \node[xshift=-8.5em,yshift=.6em] at (h1) {\scriptsize Theorem~\ref{theorem:deltacollapse}};

    \begin{pgfonlayer}{background}
%lower triangle:
%\draw[line width=3mm, draw=blue!20, fill=blue!20, rounded corners, cap=round]
%    (4-1.south) -- (4-1.south east) -- (3-4.south east) -- (3-4.east) --
%    (3-4.north east) -- (3-4.north) -- (3-4.north west) --
%    % (3-2.north) -- (3-2.north west) -- 
%    (4-1.north west) -- (4-1.west) -- (4-1.south west) -- (4-1.south);  

\end{pgfonlayer}
    
\end{tikzpicture}
\end{center}
  \caption{Hierarchy among enumeration complexity classes for all $k\geq 1$.
  The lines without `?' represent strict inclusions under some reasonable
  complexity theoretic assumption.
  For the line with `?', inclusion holds, but it is not clear whether it is strict.
  % The dotted line between indicates that the relation is still open. 
}
  \label{fig:hierarchy}
\end{figure}

The relations among the enumeration complexity classes introduced in this chapter are 
summarized in Figure~\ref{fig:hierarchy}. 

\subsection{Connections to Decision Complexity}

%We now prove several properties of these complexity classes. First, we draw a connection
%between the complexity of enumeration and
%
%decision problems. 
For a lot of problems, the set of all solutions can be enumerated by repeatedly solving
differing instances of a corresponding decision problem and using the information gained
from to construct the solutions. 

For the class $\delayP_p^{\calC}$, it turns out that the decision problem $\pExtSol{R}$
is most relevant for achieving the aforementioned goal.
Indeed, the standard enumeration algorithm that constructs solutions bit by bit and tests
whether the current partial candidate solution can still be extended to a solution 
\cite{Strozecki2010,DBLP:journals/fuin/CreignouV15}, and which outputs the solutions in 
lexicographical order, gives the following relationship.

% It turns out that in order to study the class $\delayP_p^{\calC}$ 
% the $\pExtSol{R}$ problem is most relevant.
% Indeed, the 
% standard enumeration algorithm \cite{Strozecki2010,DBLP:journals/fuin/CreignouV15}, which outputs the solutions in lexicographical order, gives the following relationship.

%
\begin{proposition}\label{prop:new}
Let $R$ be a binary relation. Further let
$k\geq 0$ and $\calC\in\{\DeltaP{k},\SigmaP{k}\}$.
If $\pExtSol{R}\in\calC$ then $\pEnum{R}\in\delayP_p^{\calC}$. 
\end{proposition}
\begin{proof}
Let $k\geq 1$. Further let $q$ be a polynomial such that for all $(x,y)\in R$, $|y|\leq q(|x|)$ holds. 
As $\delayP_p^{\DeltaP{k+1}}=\delayP_p^{\SigmaP{k}}$, we show that $\pEnum{R}\in\delayP_p^{\calC}$
by giving an enumeration algorithm $\calA$ with access to $\pExtSol{R}$ as an oracle. Assume that we have
given $x$ as an input for $\pEnum{R}$. For every $y_1\in\Sigma$, we make an oracle call with 
$(x,y_1)$ as input. If the oracle returns 'yes', $y_1$ can be extended to some $y\in R(x)$, so for every
$y_2\in\Sigma$, we query $(x,y_1y_2)$ to the oracle. Repeating this $q(|x|)$ times gives
a first answer $y\in R(x)$. With the use of backtracking, we can enumerate all of $R(x)$ with a polynomial delay and access to the oracle $\pExtSol{R}$, where every input to an oracle call is polynomially bounded by $q$.
\end{proof}

When considering decision problems, typically not the problem $\pExtSol{R}$
is studied, but instead the problem $\pExist{R}$ that just asks, given some instance,
whether there exists any solution at all. While in general unrelated, for a big class
of problems, the so-called self-reducible problems, the problem $\pExtSol{R}$ can be
solved via $\pExist{R}$.

% An important class of search problems are those for which search reduces
% to decision, the so-called self-reducible problems.
% This notion can be captured by the following definition.

%
\begin{definition}[self-reducibility]
\label{def:self-reducible}
Let $\leq_T$ denote Turing reductions.
We say that a binary relation $R$ is {\em self-reducible\/}, if
$\pExtSol{R}\leq_T\pExist{R}$.
\end{definition}

Self-reducibility thus allows to reduce enumeration to the typical decision
problem associated with $R$. More formally, for self-reducible problems,
Proposition~\ref{prop:new} can be refined as follows.

\begin{proposition}\label{prop:selfreducible}
Let $R$ be a binary relation, which is self-reducible, and
$k\geq 0$. Then $\pExist{R}\in\DeltaP{k}$
if and only if $\pEnum{R}\in\delayP_p^{\SigmaP{k}}$.
\end{proposition}

The above proposition gives a characterization of the class $\delayP_p^{\SigmaP{k}}$
in terms of the complexity of decision problems in the case of self-reducible relations.
Analogously, the notion of ``enumeration self-reducibility'' introduced by Kimelfeld and 
Kolaitis, that directly relates the search problem $\pAnotherSol{R}$ to its corresponding
decision problem \cite{DBLP:journals/tods/KimelfeldK14}, allows a characterization of the
class $\incP^{\SigmaP{k}}$.

\begin{definition}[\cite{DBLP:journals/tods/KimelfeldK14}, enumeration self-reducibility]
\label{def:enumeration-self-reducible}
A binary relation $R$ is {\em enumeration self-reducible\/} if
$\pAnotherSol{R}\leq_T\pExistAnotherSol{R}$.
\end{definition}

The following proposition generalizes two results to the higher levels of the
enumeration hierarchy,
by following analogous arguments as for the basic level $\incP$.
The first result is given by Strozecki in \cite[Prop. 2.14]{Strozecki2010},
and the second result is given by Kimelfeld and 
Kolaitis \cite[Prop. 2.2]{DBLP:journals/tods/KimelfeldK14}.

\begin{proposition}\label{prop:characterizationIncP}
Let $R$ be a binary relation and let $k\geq 0$.
\begin{itemize}
\item The functional problem $\pAnotherSol{R}$ can be solved in polynomial time with access to
a $\DeltaP{k}$-oracle if and only if $\pEnum{R}\in\incP^{\SigmaP{k}}$.
\item Let $R$ be enumeration self-reducible. Then
$\pExistAnotherSol{R}\in\DeltaP{k}$
if and only if $\pEnum{R}\in\incP^{\SigmaP{k}}$.

\end{itemize}
\end{proposition}

\section{Declarative-style Reductions}
\label{sect:first-reduction}
As far as we know, only a few kinds of reductions between enumeration problems have been investigated so far. One such reduction is implicitly described in \cite{DurandG07}. 
It establishes a bijection between sets of solutions.
Different approaches introduced in \cite{BraultBaron13} and in \cite{Mary13} relax this condition and allow for
non-bijective reduction relations. 
We go further in that direction in proposing a declarative style reduction
relaxing the isomorphism requirement while 
closing the relevant enumeration classes.

%\todo[inline]{NC generalized the original definition so that   a polynomial number of $y\in R(\sigma(x))$,
%can be mapped  to an empty set. The formulation can probably be improved}

\begin{definition}[e-reduction]
\label{definition:firstreduction}
Let $R_1,R_2\subseteq \Sigma^*$ be binary relations. Then 
$\pEnum{R_1}$ reduces to $\pEnum{R_2}$ via an $e$-reduction,
$\pEnum{R_1} \leq_e \pEnum{R_2}$, 
if there exist
a function $\sigma:\Sigma^*\rightarrow\Sigma^*$
computable in polynomial time
and a relation $\tau\subseteq\Sigma^*\times\Sigma^*\times\Sigma^*$, s.t.\ 
for all $x\in\Sigma^*$ the following holds.
For $y\in\Sigma^*$, let $\tau(x,y,-):=\{z\in \Sigma^*\mid (x,y,z)\in\tau\}$
and for $z\in\Sigma^*$, let $\tau(x,-,z):=\{y\in \Sigma^*\mid (x,y,z)\in\tau\}$.
Then: 
\begin{enumerate}
  \item ${R_1}(x)=\bigcup_{y\in {R_2}(\sigma(x))}\tau(x,y,-)$;
   \item 
   For all $ y\in {R_2}(\sigma(x))$,  either $\tau(x,y,-)=\emptyset$, or 
   $\emptyset \subsetneq \tau(x,y,-) \subseteq {R_1}(x)$ and 
   $\tau(x,y,-)$ can be enumerated with polynomial delay in $|x|$; 
moreover $\tau(x,y,-)=\emptyset$ can only hold for a number of $y$'s which is polynomially bounded in $|x|$;
  \item  For all $z\in  {R_1}(x)$, we have  
  $\tau(x,-,z) \subseteq {R_2}(\sigma(x))$ and 
    the size of $\tau(x,-,z)$ is polynomially bounded in $|x|$.
\end{enumerate}
\end{definition}
\begin{figure}[t]
\centerline{
\begin{tikzpicture}
	\node (h1) at (-2,2.2) {${R_1}(x)$};
	\node[circle,draw,minimum width=.1cm] (d1) at (-2,1) {};
	\node[circle,draw,minimum width=.1cm,dashed] (d2) at (-2,1.4) {};
	\node[rotate=90] (dots1) at (-2,.5) {\dots};
	\node[circle,draw,minimum width=.1cm] (di) at (-2,0) {};
	\node (h2) at (1,2.2) {${R_2}(\sigma(x))$};
	\node[circle,draw,minimum width=.1cm, dotted] (c1) at (1,1.4) {};
	\node[circle,draw,minimum width=.1cm] (c2) at (1,1) {};
	\node[rotate=90] (dots2) at (1,.5) {\dots};
	\node[circle,draw,minimum width=.1cm] (ci) at (1,0) {};
	\node[ellipse,draw,rotate=90,minimum width=2.3cm, minimum height=.8cm] (e1) at (-2,.7) {};
	\node[ellipse,draw,rotate=90,minimum width=2.3cm, minimum height=.8cm] (e2) at (1,.7) {};
	\draw[-stealth', dotted, bend right] (c1) to (d1);
	\draw[-stealth', dotted, bend right] (c1) to (di);
	\node[anchor=west] (t1) at (2.2,1.4) {\begin{minipage}{7cm}
    One solution of $R_2(\sigma(x))$ may map to an unbounded number of solutions of $R_1(x)$.\end{minipage}};
  \node[circle,draw,minimum width=.1cm, anchor=south west,xshift=.1cm,yshift=.1cm] (il1) at (t1.north west) {};
  \node[circle,draw,minimum width=.1cm, xshift=2.5cm, dotted] (il2) at (il1) {};
  \draw[-stealth', dotted] (il2) to (il1);  
	\draw[-stealth', dashed, bend right] (c2) to (d2);
	\draw[-stealth', dashed, bend right] (ci) to (d2);
	\node[anchor=west] (t2) at (2.2,-.1) {\begin{minipage}{7cm}At most polynomially
	many solutions of $R_2(\sigma(x))$ may map to one solution of $R_1(x)$.\end{minipage}};
  \node[circle,draw,minimum width=.1cm, anchor=south west,xshift=.1cm,yshift=.1cm, dashed] (il3) at (t2.north west) {};
  \node[circle,draw,minimum width=.1cm, xshift=2.5cm] (il4) at (il3) {};
  \draw[-stealth', dashed] (il4) to (il3);  
	 \node[] at (d1) {$z$};
	\node[] at (c2) {$y$};
\end{tikzpicture}
}
  \caption{Illustration of relation $\tau$ from 
  Definition~\ref{definition:firstreduction}.}
  \label{fig:reduction-with-tau}
\end{figure}
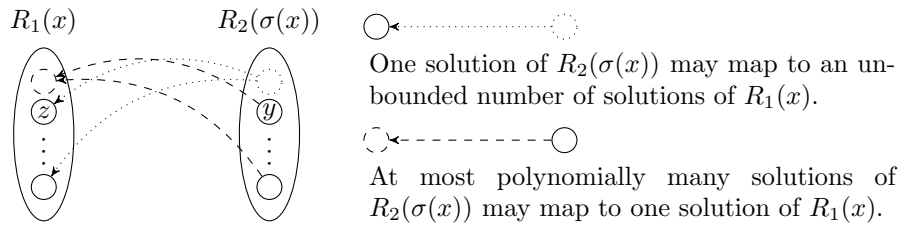

\medskip
Intuitively, $\tau$ 
establishes a relationship between instances 
$x$, solutions $y \in {R_2}(\sigma(x))$ and 
solutions $z \in {R_1}(x)$.
We can thus use $\tau$ to design an enumeration algorithm for ${R_1}(x)$
via an enumeration algorithm for ${R_2}(\sigma(x))$.
The conditions imposed on $\tau$ have the following meaning: 
By condition 1, the solutions $z \in {R_1}(x)$ can be computed by 
iterating through the solutions $y \in {R_2}(\sigma(x))$ and 
computing $\tau(x,y,-) \subseteq {R_1}(x)$.
Conditions 2 and 3 make sure that the delay of enumerating  ${R_1}(x)$ 
only differs by a polynomial from the delay of enumerating ${R_2}(\sigma(x))$: 
Condition 2 ensures that, for every $y$, 
the set $\tau(x,y,-)$ can be enumerated with polynomial delay and
%  that we 
% never encounter a ``useless'' $y$ (i.e., a solution $y \in \mathsf{Sol}_{R_2}(\sigma(x))$ 
% which is associated with 
% no solution $z \in \mathsf{Sol}_{R_1}(x)$). 
 that we 
do not  encounter more than a polynomial number  of ``useless'' $y$ (i.e., a solution $y \in {R_2}(\sigma(x))$ 
which is associated with 
no solution $z \in {R_1}(x)$).

In principle, we may thus get duplicates $z$  associated with different values of $y$. 
However, Condition 3 ensures that each $z$ can be associated with at most
polynomially many values $y$. Using a priority queue storing all $z$ that are output,
we can avoid duplicates, cf.\ the proof of Proposition~\ref{prop:ereductionsComposeClosed}
or~\cite{Strozecki2010}.
Figure~\ref{fig:reduction-with-tau} illustrates the idea of $\tau$, and a concrete
example of an e-reduction is provided next.
\begin{example}
The idea of the relation $\tau$ can also be nicely demonstrated on an $e$-reduction
from $\threecolE$ to $\fourcolE$ (enumerating all valid 3- respectively
4-colourings of a graph). We intentionally choose this reduction since there is
no bijection between the solutions of the two problems. 

 Recall the classical many-one reduction between these problems, which takes a graph
 $G$ and defines a new graph $G'$ by adding an auxiliary vertex $v$ and connecting it
 to all the other ones. This reduction can be extended to an $e$-reduction with
 the following relation $\tau$:
 With every graph $G$ in the first component of $\tau$, we associate all valid 
 4-colourings (using 0, 1, 2, and 3) of $G'$ in the third component of $\tau$.
 With each of those
 we associate the corresponding
 3-colouring of $G$ in the second component.
 They 
 are obtained from
 the 4-colourings by first making sure that $v$ is coloured with $3$ (by ``switching''
 the colour of $v$ with 3) and then by simply reading off the colouring of the
 remaining vertices.
\end{example}

%The following lemma relates e-reductions to many-one reductions. 
%
%\todo[inline]{NC Lemma \ref{prop:firstredmanyone} does not   hold  anymore with the generalization I have given for the reduction, and actually it is good. Indeed this lemma says that under such reduction a hard enumeration problem corresponds necessarily to a hard decision problem, which is not always the case, some solutions can be easy to find. So the first reduction was much too strong.}
%
%\begin{lemma}\label{prop:firstredmanyone}
%Let $R_1,R_2$ be binary relations and let $\leq_m$ denote many-one reductions.
%Then we have:  If $\pEnum{R_1}\leq_e\pEnum{R_2}$,
%then $\pExist{R_1}\leq_m\pExist{R_2}$.
%\end{lemma}
%\begin{proof}
%Let $R_1,R_2$ be binary relations with $\pEnum{R_1}\leq_e \pEnum{R_2}$ and let $\sigma\in\mathsf{FP}$ and
%$\tau\in\Sigma^*\times\Sigma^*\times\Sigma^*$ relations witnessing that this reduction holds.
%It suffices to show that for all $x\in\Sigma^*$, $x\in\pExist{R_1}$ if and only if $\sigma(x)\in\pExist{R_2}$.
%So fix some $x$ and assume that $x\in\pExist{R_1}$. Then ${R_1}(x)$ is nonempty, and since
%${R_1}(x)=\bigcup_{y\in{R_2}(\sigma(x))}\tau(x,y,-)$ also ${R_2}(\sigma(x))$ is nonempty,
%and thus $\sigma(x)\in\pExist{R_2}$. For the other direction, assume that $\sigma(x)\in\pExist{R_2}$.
%This means that there exists some $y\in\Sigma^*$ with $y\in{R_2}(\sigma(x))$. Since
%$\tau(x,y,-)$ is nonempty by definition and $\tau(x,y,-)\subseteq{R_1}(x)$ we are done.
%\end{proof}

The $e$-reductions have two desirable and important  properties, as stated
next.
\begin{proposition}\label{prop:ereductionsComposeClosed}
\begin{enumerate}
	\item Reducibility via $e$-reductions is a transitive relation.
	\item Let $\calC \in \{\SigmaP{k}, \DeltaP{k} \mid k \geq 0\}$. The classes $\delayP_p^\calC$, $\delayP^\calC$, and $\incP^\calC$ are closed under $e$-reductions. 
\end{enumerate}
\end{proposition}

\begin{proof}
We first show that  the enumeration classes $\delayP^\calC$, $\delayP_p^\calC$ and $\incP^\calC$
are closed under $e$-reductions for $\calC \in \{\SigmaP{k}, \DeltaP{k} \mid k \geq 0\}$.
Let $R_1,R_2$ be binary relations with $\pEnum{R_1}\leq_e\pEnum{R_2}$. Further let $\sigma$ and $\tau$
be relations corresponding to the reduction $\pEnum{R_1}\leq_e\pEnum{R_2}$, and assume that
 $\pEnum{R_2}\in\delayP^\calC$ (the cases where
$\pEnum{R_2}\in\delayP_p^\calC$ or $\pEnum{R_2}\in\incP^{\calC}$ work along the same lines).
Let $\calA$ denote the enumeration
algorithm for $\pEnum{R_2}$ with a polynomial delay and decision oracle $\calC$, and let
$\calB$ be the polynomial delay algorithm enumerating $\tau(x,y,-)$ for all $x,y\in\Sigma^*$. Moreover,
let $p$ be a polynomial such that for all $z\in{R_1}(x)$, we have $|\tau(x,-,z)|\leq p(|x|)$. The idea for
an enumeration algorithm for $\pEnum{R_1}$ is to enumerate (without output) 
${R_2}(\sigma(x))$ via $\calA$, and for every element $y$ that would be output by $\calA$,
repeatedly add $p(|x|)$ elements of $\tau(x,y,-)$ to a priority queue. Then, whenever those elements
are added to the queue, an element of the queue w.r.t. some order is output. This way, one can ensure
polynomial (respectively incremental) delay albeit producing an exponentially large priority queue.\\
To give a detailed explanation of the enumeration algorithm for $\pEnum{R_1}$, fix some $x\in\Sigma^*$.
Denote by $\mathsf{Newoutput}_\calA(\sigma(x))$ a new output made by the enumeration algorithm $\calA$
when enumerating ${R_2}(\sigma(x))$, and
similarly by $\mathsf{Newoutput}_\calB(y)$ a new output made by the enumeration algorithm $\calB$
when  enumerating $\tau(x,y,-)$. Algorithm~\ref{algorithm:reductiononeclosed} gives the algorithm
 for enumerating ${R_1}(x)$. It
is easy to see that this algorithm indeed works with a polynomial (respectively incremental) delay. When adding
some element $z$ to the priority queue (line~\ref{line:addqueue} of Algorithm~\ref{algorithm:reductiononeclosed}), if
$z$ is already in the queue, then $\mathsf{output\_queue}$ remains unchanged, otherwise $z$ is added as the last
element. Moreover, note that we do not delete elements in the priority queue in order to avoid the duplicates in
the output.

\begin{algorithm}
\caption{Enumerate ${R_1}(x)$}
\label{algorithm:reductiononeclosed}
\begin{algorithmic}[1]
\State $i = 1$
\State $\mathsf{output\_queue}=\emptyset$
\\
\While{ $\calA$ has not output all of ${R_2}(\sigma(x))$}
 
\State $y =\mathsf{Newoutput}_\calA(\sigma(x))$
\State $\mathsf{output\_delay} = 1$

\While{ $\mathsf{output\_delay}< p(|x|)$ }
\State $z=\mathsf{Newoutput}_\calB(y)$
\If{ $z\neq \emptyset$}
\State Add $z$ to $\mathsf{output\_queue}$\label{line:addqueue}
\State $\mathsf{output\_delay} = \mathsf{output\_delay} +1$
\Else
\State $y =\mathsf{Newoutput}_\calA(\sigma(x))$
\EndIf
\If {$y = \emptyset$}
\State Output (without deleting) the elements of $\mathsf{output\_queue}$ starting from the $i$-th element
\State $\mathsf{output\_delay}=p(|x|)$
\EndIf

\EndWhile
\If {$y \neq \emptyset$}
\State Output (without deleting) the $i$-th element of $\mathsf{output\_queue}$
\EndIf
\State $i = i+1$
\EndWhile
\end{algorithmic}
\end{algorithm}

To show that reducibility via $e$-reductions is a transitive relation,
let $R_1,R_2,R_3\subseteq \Sigma^*$ be binary relations.
Further let $\pEnum{R_1}\leq_e\pEnum{R_2}$ with corresponding polynomial $p_1$ and relations
$\tau_1$ and $\sigma_1$
and $\pEnum{R_2}\leq_e\pEnum{R_3}$  with corresponding polynomial $p_2$ and
 relations $\tau_2$ and $\sigma_2$.
Define relations $\sigma_3$ as $\sigma_3:=\sigma_2\circ\sigma_1$ and
$\tau_3$ as
\begin{equation*}
\tau_3:=\{(x,y,z)\in\Sigma^*\times\Sigma^*\times\Sigma^*\mid\exists\zeta\in\Sigma^*\text{ with }(x,\zeta,z)\in\tau_1
\text{ and }(\sigma_1(x),y,\zeta)\in\tau_2\}.
\end{equation*}
Let $x,y\in\Sigma^*$. To show that $\tau_3(x,y,-)$ can be enumerated with a polynomial delay in $|x|$, we use
the same idea for a polynomial delay enumeration as we did to show that the reduction closes the enumeration classes: First a single $y'\in {R_2}(\sigma_1(x))$
is computed by $\tau_2(\sigma_1(x),y,-)$ (with a polynomial delay in $\sigma_1(|x|)$ and thus polynomial
delay in $|x|$), and then at most polynomially many $y''$ from
$\tau_1(x,y',-)$ are added to some priority queue. An element of the queue is output, and again 
polynomially many elements are added by the queue (possibly by first computing some new 
$y'\in {R_2}(\sigma_1(x))$. This way, all of $\tau_3(x,y,-)$ can be enumerated with a polynomial
delay in $|x|$.
Moreover, it is easy to see that for all $x,z\in\Sigma^*$ we have $|\tau_3(x,-,z)|\leq p_1(p_2(|x|))$ and
that ${R_2}(x)=\bigcup_{y\in{R_3}(\sigma_3(x))}\tau_3(x,y,-)$. It follows
that indeed $\pEnum{R_1}\leq_e\pEnum{R_3}$.
\end{proof}

Note that, as an additional benefit, $\pEnum{R_1}\leq_e\pEnum{R_2}$
does not imply $\pExist{R_1}\leq_m\pExist{R_2}$. An example for this
behavior can be found in Theorem~\ref{thm:dichotomy} below.

Thus, besides being a natural extension of existing reductions,
Proposition~\ref{prop:ereductionsComposeClosed} also shows that
$e$-reductions have important desirable properties.
As a result, they are very well suited for reasoning about enumeration
algorithms and complexities. For example, they allow to compare different
enumeration problems, thus supporting to make statements on one problem
``relative'' to some other problem, and can be used to derive enumeration
algorithms -- hence complexity upper bounds.

In fact, they have already been implicitly used previously for these purposes.
For example, an important result from \cite{KanteLMN14} compares the complexity
of two well-known problems by implicitly using an $e$-reduction:
\begin{example}
We denote by $\domenum$ the problem of enumerating all subset-minimal dominating sets of a given graph,
and by $\transenum$ the problem of enumerating all minimal hitting sets of a hypergraph. 
In~\cite{KanteLMN14}, it was shown that
$\transenum$ can be enumerated via an enumeration algorithm for $\domenum$, by mapping a hypergraph to an 
instance of $\domenum$, and giving a (non-bijective) map between the minimal dominating sets and the minimal
hitting sets. In fact this defines the following $e$-reduction $\transenum\leq_e\domenum$:
\begin{itemize}
\item Given a hypergraph $\calH$, $\sigma(\calH)$ is defined as the co-bipartite incidence
graph of $\calH$, with
\begin{align*}
V(\sigma(\calH))=&V(\calH)\cup\{y_e\mid e\in E(\calH)\}\cup\{v\},\\
E(\sigma(\calH))=&\{\{x,y_e\}\mid x\in V(\calH),x\in e\}\cup\{\{x,y\}\mid x,y\in V(\calH)\}\\
&\cup \{\{v,x\}\mid x\in V(\calH)\}\cup\{\{y_e,y_f\}\mid e,f\in E(\calH)\}.
\end{align*}
\item Given $\sigma(\calH)$ and some minimal dominating set $D\subseteq\sigma(\calH))$, we have that
\begin{equation*}
\tau(\sigma(\calH), D,-)=\left\{
     \begin{array}{ll}
       \emptyset & : D=\{x,y_e\}\text{ for some }x\in V(\calH),e\in E(\calH),\\
       D & :\text{ otherwise}.
     \end{array}
   \right.
\end{equation*}
\end{itemize}
\end{example}
Note that $\domenum\leq_e\transenum$ holds as well.
The original question studied in \cite{KanteLMN14} was whether any of these problems
can be enumerated in a time that is polynomially in both, the size of the input and
the size of the output. For the two problems $\domenum$ and $\transenum$, this is a
longstanding open problem. The conclusion in \cite{KanteLMN14} was that
% the two problems are equivalent under output-polynomial time, i.e. 
if such an
enumeration algorithm exists for one of the problems, then this is also possible for
the other one. 
It can be shown that the existence of the $e$-reductions between $\domenum$ and $\transenum$
also induces the same result.
% the equivalence of the two enumeration problems under output-polynomial time.
Moreover, they are also equivalent for all of the classes that are closed under
$e$-reductions as given in Proposition~\ref{prop:ereductionsComposeClosed},
i.e., if one of the two problems is in any of these classes, then the other
problem is as well.

%By the properties of e-reductions, not only does this result follow immediately from
%the existence of such a reduction, but equivalent results for all of the classes that
%are closed under e-reductions.

Another example that implicitly uses $e$-reductions is \cite{creignou1997generating}.
In later sections of this paper, we will use them to retrieve reductions
among enumeration problems as well. And even though we do not show any completeness
results under $e$-reduction, we will actually use them to show hardness under alternative
reductions that will be introduced in the next section.

\section{Procedural-style Reductions}
% and Completeness Results}
\label{sect:second-reduction}

Although Turing
reductions are too strong to show completeness results for
classes in the polynomial hierarchy, because the classes $\SigmaP{k}$ and $\DeltaP{k+1}$ have the same closure, Turing style reductions
turn out to be meaningful in our case. In this section we 
introduce two types of reductions that are motivated by Turing
reductions. Both of them are able to reduce
between enumeration problems for which the 
$e$-reductions
seem to be too weak.

Towards this goal, we first have to define the concept of 
RAMs with an \emph{oracle for enumeration problems} (such oracles were already introduced in \cite{Strozecki2010}).
The intuition behind the definition of such enumeration
oracle machines is the following: For algorithms (i.e., Turing 
machines or RAMs in the case of enumeration) using a decision
oracle for the language $L$, we usually have a special instruction
that given an input $x$ decides in one step whether $x\in L$, and 
then executes the next step of the algorithm accordingly. For an 
algorithm $\calA$ using an enumeration oracle, an input $x$ to some
$\pEnum{R}$-oracle returns in a single step (using the instruction
$\newoutput$, see the definition below) a single element of
$R(x)$, and then 
$\calA$  can proceed according to this output.

\begin{definition}[Enumeration Oracle Machines]\label{def:eoma}
 Let $\pEnum{R}$ be an enumeration problem. 
 An \emph{Enumeration Oracle Machine with an enumeration oracle
 $\pEnum{R}$ (EOM\_R)} is a RAM
 that has a sequence of new registers $A_e, O^e(0),O^e(1),\ldots$
 and a new instruction $\newoutput$ (next Oracle output). 
 An EOM\_R 
 is \emph{oracle-bounded} if the size of all
 inputs to the oracle is at most polynomial in the size of
 the input to the EOM\_R.
\end{definition}

When executing $\newoutput$, the machine writes -- in one step --
some $y_i \in R(x)$ to register $A_e$, where
$x$ is the word stored in $O^e(0),O^e(1),\ldots$ and $y_i$ is
defined as follows:

\begin{definition}[Next Oracle Output] \label{def:nop}
 Let $R$ be a binary relation, $\pi_1, \pi_2, \dots $ be the
 run of an EOM\_R and assume that the $k^{th}$ instruction
 is $\newoutput$, i.e., $\pi_k = \newoutput$.
 Denote with $x_i$ the word stored in $O^e(0), O^e(1), \dots$
 at step $i$.
 Let $K  = \{ \pi_i \in \{\pi_1, \dots, \pi_{k-1}\} \mid 
 \pi_i = \newoutput \text{ and } x_i = x_k\}$. 
 Then the \emph{oracle output $y_k$ in $\pi_k$} is defined as
 an arbitrary $y_k \in R(x_k)$ s.t.\
 $y_k$ has not been the oracle output in any $\pi_i \in K$.
 If no such $y_k$ exists, then the oracle output in $\pi_k$
 is undefined.

 When executing $\newoutput$ in step $\pi_k$, if the oracle
 output $y_k$ is undefined, then the register $A_e$ contains
 some special symbol in step $\pi_{k+1}$. Otherwise in step
 $\pi_{k+1}$ the register $A_e$ contains $y_k$ .
\end{definition}
Observe that since an EOM $M^e$ is a polynomially bounded
RAM and the complete oracle output is stored in the register $A_e$,
only such oracle calls are allowed where the size of each
oracle output is guaranteed to be polynomially bounded in the size
of the input of $M^e$.

Using EOMs, we can now define another type of reductions among
enumeration problems, reminiscent of classical Turing reductions.
I.e., we say that one problem $\pEnum{R_1}$ reduces to another
problem $\pEnum{R_2}$, if $\pEnum{R_1}$ can be solved by an EOM
using $\pEnum{R_2}$ as an enumeration oracle.
\newpage
\begin{definition}[$D$-reductions, $I$-reductions]\label{def:turingReductions}
 Let $R_1$ and $R_2$ be binary relations.
 
 \begin{itemize}
  \item We say that
  $\pEnum{R_1}$ reduces to $\pEnum{R_2}$ via $D$-reductions, 
   $\pEnum{R_1}\leq_D\pEnum{R_2}$, if there
  is an 
 oracle-bounded  EOM\_$R_2$ that enumerates $R_1$ 
  in $\delayP$
  and is independent of the order in which the $\pEnum{R_2}$-oracle
  enumerates its answers.

  \item We say that 
  $\pEnum{R_1}$ reduces to $\pEnum{R_2}$ via $I$-reductions,   
  $\pEnum{R_1}\leq_I\pEnum{R_2}$, if there 
  is an 
  EOM\_$R_2$ that enumerates $R_1$ in
  $\incP$ 
  and is independent of
  the order in which the $\pEnum{R_2}$-oracle enumerates its
  answers.
 \end{itemize}
\end{definition}

The following proposition shows that these reductions have desirable properties:

\begin{proposition} \label{prop:turingReductionsCompose}
Reducibility via $D$-reductions as well as reducibility via $I$-reductions are
transitive relations.
\end{proposition}

\begin{proof}
This can be proven along the same lines as Proposition~\ref{prop:reductionclosed} by substituting
occurrences of enumeration RAMs with decision oracles by enumeration RAMs
\end{proof}

For $D$-reductions, we require the EOM\_$R_2$ to be oracle-bounded.
We would like to point out that this restriction is essential:
if we drop it, then the classes $\delayP_p^\calC$ are not closed
under the resulting reduction. Indeed, suppose that the $D_u$-reduction is defined similar to the $D$-reduction but
without the restriction to an oracle-bounded EOM. Then, 
abusing the decision problem $\ase$ from the proof of Theorem~\ref{theorem:deltacollapse} as an
enumeration problem $\pEnum{R_A}$ (with output 'yes' or 'no'), we can show that for any $k\geq1$ and any
$\pEnum{R}\in\incP^{\SigmaP{k}}$ we have that $\pEnum{R}\leq_{D_u}\pEnum{R_A}$ but
$\pEnum{R_A}\in\delayP^{\SigmaP{k}}_p$.

\begin{proposition}\label{prop:reductionclosed}
Let $\calC \in \SigmaP{k}, k \geq 0$. 
The class $\delayP_p^{\calC}$ is
closed under $D$-reductions and
the class $\incP^{\calC}$ is closed under $I$-reductions.
\end{proposition}

\begin{proof}
 Let $M$ be an oracle-bounded enumeration oracle machine with an enumeration oracle $\pEnum{R_2}$
witnessing that $\pEnum{R}\leq_D\pEnum{R_2}$. Let $\calA$ be the polynomial delay algorithm with
access to a polynomially bounded $\calC$-oracle We can construct a
RAM $M'$ that enumerates $\pEnum{R_1}$ with a polynomial delay using a polynomially bounded
decision oracle $\calC$, by
modifying the RAM $M$ as follows: Every time $M$ makes a call to an $\pEnum{R_2}$-oracle, we use 
the algorithm $\calA$ to retrieve what should be written to the register $A_e$. Assume that $x$ is the input to 
an oracle call of the RAM $M$. Then the new RAM $N$ assigns two fixed addresses $a_x^0$ and
$a_x^1$ to $x$. Then $N$ can simulate the algorithm $\calA$ on the registers $R(2^{a_x^0}),\ldots,R(2^{a_x^1})$
until $\calA$ would output some $y\in\Sigma^*$. The RAM $N$ writes $y$ to $A_e$, and a simulation
of a single oracle call is completed.
Whenever $x$ is the input for a $\newoutput$-call, $N$ continues to simulate $\calA$ on those registers; this way,
the enumeration of ${R_2}(x)$ does not need to start from the beginning every time $x$ is the input
of an oracle call. 
The proof of the closure of $\incP^\calC$ under $I$-reductions can be done along the same lines.
\end{proof}

We note that all of these properties still hold when there
is no oracle at all, i.e., for the classes 
$\delayP$ and $\incP$.
\medskip

Observe that the $e$-reductions introduced in the previous section are
a particular case of $D$- and $I$-reductions. We will also use $e$-reductions later when we establish
completeness results for specific problems.

%\todo[inline]{NC added this proposition, is the proof clear enough?} 

\begin{proposition} \label{prop:reduction_E_compared_to_D}
 Let $R_1$ and $R_2$ be binary relations.
\begin{itemize}
\item If $\pEnum{R_1}\leq_e\pEnum{R_2}$, then $\pEnum{R_1}\leq_D\pEnum{R_2}$.
\item If $\pEnum{R_1}\leq_D\pEnum{R_2}$, then $\pEnum{R_1}\leq_I\pEnum{R_2}$.
\end{itemize}
\end{proposition}

\begin{proof} The proof of the second claim follows immediately from the definitions.
Now, if $\pEnum{R_1}\leq_e\pEnum{R_2}$, then  Algorithm \ref{algorithm:reductiononeclosed} developed in the proof of Proposition~\ref{prop:ereductionsComposeClosed} provides an oracle bounded EOM\_$R_2$ that enumerates $R_1$ in $\delayP$. 
For enumerating $R_1(x)$ the calls to the enumeration oracle  $\pEnum{R_2}$ are only made on the input $\sigma(x)$, where $\sigma$ is the function used in the $e$-reduction.
\end{proof}

Now, unlike for $e$-reductions, the next theorem shows that the
$D$- and $I$-reduction induce complete
problems for the enumeration complexity classes introduced in
Section~\ref{sect:classes}.

\begin{theorem}\label{theorem:selfreducomplete}
  Let $R$ be a binary relation and $k\geq1$ such that
  $\pExist{R}$ is $\SigmaP{k}$-hard.
  \begin{itemize}    
   \item $\pEnum{R}$ is $\delayP_p^{\SigmaP{k}}$-hard via
      $D$-reductions.
	\item $\pEnum{R}$ is $\incP^{\SigmaP{k}}$-hard via $I$-reductions.
   \item If additionally $\pExist{R}$ is in $\SigmaP{k}$ and $R$ is self-reducible, then $\pEnum{R}$ is
    $\delayP_p^{\SigmaP{k}}$-complete via $D$-reductions and
    $\incP^{\SigmaP{k}}$-complete via $I$-reductions.
  \end{itemize}
\end{theorem}
% \begin{proof}[Idea]
% Let $\pEnum{R'}\in\delayP_p^L$ for some $L\in\SigmaP{k}$, and assume that $z$ is the input to
% an $L$-oracle when enumerating ${R'}(x)$ for some $x\in\Sigma^*$. As $\pExist{R}$ is
% $\SigmaP{k}$-complete and the enumeration is oracle-bounded, $z$ can be transformed to
% an equivalent instance $z'$ of $\pExist{R}$ in time polynomial only in $|x|$.
% Therefore by calling the $\pEnum{R}$-oracle once
% and by checking whether ${R}(z')=\emptyset$, one can decide whether $z\in L$.
% The membership $\pEnum{R}\in\delayP_p^{\SigmaP{k}}$ in the case of self-reducibility
% follows by Proposition~\ref{prop:new}.
% \end{proof}

\begin{proof}
Let $R$ be a relation such that $\pExist{R}$ is $\SigmaP{k}$-complete for some $k\geq 1$. 
\begin{itemize}
\item We have to prove that for any $\pEnum{R'}\in\delayP_p^{\SigmaP{k}}$, $\pEnum{R'}\leq_D\pEnum{R}$
($\incP^{\SigmaP{k}}$-hardness via $I$-reductions can be shown along the same lines).
So let $R'$ be a binary relation
such that $\pEnum{R'}\in\delayP_p^{\SigmaP{k}}$. By definition there is some $L\in\SigmaP{k}$ such that
$\pEnum{R'}\in\delayP_p^L$. Moreover let $\calA$ be an algorithm witnessing this membership.
As $\pExist{R}$ is $\SigmaP{k}$-complete, we have that
$L\leq_m \pExist{R}$, so any input $x$ to an $L$-decision oracle when enumerating $\pEnum{R}$ can be
transformed to an instance $x'\in\pExist{R}$ such that $x\in L$ iff $x'\in\pExist{R}$, and this transformation
can by done in polynomial time in the size of $x$. Moreover, since the size of the oracle input is polynomial,
this reduction can be computed within the time bounds of a polynomial delay, i.e.
whenever a polynomial delay algorithm with an $L$-oracle makes an oracle call with an input $x$, the same
algorithm can also perform a transformation to some $x'$ before that oracle call, without violating
the polynomial delay restriction. Therefore we can enumerate $\pEnum{R'}$ with an oracle bounded
enumeration oracle machine with $\pEnum{R}$ as follows: Whenever $\calA$ would make a decision
oracle call to $L$ with input $x$, instead the machine transforms this to some $x'\in\Sigma^*$, and then
makes a $\newoutput$-instruction with input $x'$ to the $\pEnum{R}$-oracle. The $\newoutput$-instruction
writes a nonempty string to the register $A_e$ if and only if $x'\in\pExist{R}$ and thus if and only if $x\in L$. 
It follows that we can simulate the decision oracle call with an enumeration oracle call.
\item Membership of $\pEnum{R}$ in $\delayP_p^{\SigmaP{k}}$ in this case follows
immediately from Proposition~\ref{prop:new}. Since $\delayP_p^{\SigmaP{k}}\subseteq\incP^{\SigmaP{k}}$, this
also shows membership in $\incP^{\SigmaP{k}}$.
\end{itemize}
\end{proof}

As a consequence, the enumeration problems $\sigmaSATe{k}$ and also $\piSATe{k}$
are natural complete problems for our enumeration complextiy classes:

%\todo[inline]{NC Why don't we state  in this corollary that $SAT^e$ is also $\delayP_p^{\NP}$-complete? it follows from the previous thm doesn't it}

\begin{corollary}\label{cor:complexity_sat}
 Let $k\geq 1$. Then
 \begin{enumerate}

  \item $\sigmaSATe{k}$ is complete for $\delayP_p^{\SigmaP{k}}$
    via $D$-reductions.

  \item $\piSATe{k}$ and $\sigmaSATe{k+1}$ are complete for 
    $\incP^{\SigmaP{k+1}}$ via $I$-reductions. 
In particular $\SATe$ the  enumeration variant of the traditional {\rm SAT} problem is 
$\incP^{\NP}$-complete via $I$-reductions.
 \end{enumerate}
\end{corollary}

\begin{proof}
The results for $\sigmaSATe{k}$  follow immediately from Theorem \ref{theorem:selfreducomplete}. It only remains to prove that
  that $\piSATe{k}$ and $\sigmaSATe{k+1}$
are equivalent under $I$-reductions. 
Note that $\piSATe{k}\leq_I\sigmaSATe{k+1}$ follows
immediately from the fact that $\piSATe{k}$ is a special case of $\sigmaSATe{k+1}$, so it suffices to show 
that $\sigmaSATe{k+1}\leq_I\piSATe{k}$.
Thus consider an instance $\psi$ of $\sigmaSATe{k+1}$ given
as $\psi(x):=\exists y_0\forall y_1\ldots Q_k y_k \phi(x,y_0,\ldots,y_k)$.
We can enumerate all solutions to $\psi$ as follows:
The first input to a $\piSATe{k}$-oracle is $\psi_0(x,y_0):=\forall y_1\ldots Q_k y_k \phi(x,y_0,\ldots,y_k)$,
with free variables $x$ and $y_0$. A single $\newoutput$ instruction thus gives a solution
$x_0,y'_0$ for 
$\psi_0$, and $x_0$ can be output as a solution to $\psi$. The next solution can
be found by calling a $\newoutput$ instruction with the input 
$\psi_1(x,y_0)=\forall y_1\ldots Q_k y_k (\phi(x,y_0,\ldots,y_k)\wedge(x_0\neq x))$. We only
need to add the clauses of $(x_0\neq x)$ to the input registers of the oracle tape, and we can choose an
encoding such that this does not alter the previous input, but extends it. The output $x_1,y''_0$ of the
second oracle call gives the second output $x_1$ for the $\sigmaSATe{k+1}$ problem.
By repeating this method until an oracle call gives back the empty solution, we can enumerate the solutions
of $\psi$.
\end{proof}

Observe that, via different reductions, $\sigmaSATe{k}$ is 
complete for both, $\incP^{\SigmaP{k}}$ and for the 
presumably smaller class $\delayP_p^{\SigmaP{k}}$. This provides
additional evidence that the two reductions nicely capture 
$\incP^{\SigmaP{k}}$ and $\delayP_p^{\SigmaP{k}}$,
respectively. 
%Also from Corollary~\ref{cor:complexity_sat} it follows
%as a special case that 
%$\incP^{\SigmaP{0}}$ and $\incP^{\SigmaP{1}}$ are equivalent
%under $\leq_I$ reductions:
%Clearly, $\sigmaSATe{0} = \piSATe{0}$, since in both cases the
%formul\ae are quantifier free and one asks for all satisfying
%truth assignments. Now by the theorem we know that both,
%$\sigmaSATe{1}$ and $\piSATe{0}$, and thus also $\sigmaSATe{0}$,
%are complete for $\incP^{\SigmaP{1}}$. As a result we have
%that the enumeration variant of the traditional SAT problem
%is $\incP^\NP$-complete.
%\todo{isn't it also  $\delayP_p^\NP$ complete under D-reductions? This paragraph above has to be rewritten}
\medskip

Roughly speaking Theorem~\ref{theorem:selfreducomplete} says that any  enumeration problem whose corresponding decision problem is hard, is hard as well. An interesting question is whether there exist easy decision problems for which the corresponding enumeration problem is hard. We answer positively to this question in the next section.

\section{Completeness Results}\label{sect:completeness}

In this section, we prove completeness results under procedural-style reductions for several 
problems. We start by considering generalized satisfiability in Schaefer's framework, and
then continue by looking into minimal models of Boolean formul\ae.
Next, we study enumeration within the framework of abduction, and finally turn to the
problem of enumerating repairs of an inconsistent database.

\subsection{Generalized Satisfiability}

 In this subsection, we revisit   a classification theorem obtained for the enumeration of generalized satisfiability \cite{creignou1997generating} in our framework. 
It is convenient to first introduce some notation.

A \emph{logical relation} of arity $k$ is a relation $R\subseteq\{0,1\}^k$. 
A \emph{constraint}, $C$, is a formula $C=R(x_1,\dots,x_k)$,
where $R$ is a logical relation of arity $k$ and the $x_i$'s are variables.
An assignment $m$ of truth values to the variables \emph{satisfies}
the constraint $C$ if $\bigl(m(x_1),\dots,m(x_k)\bigr)\in R$. 
A \emph{constraint
language} $\Gamma$ is a finite set of  nontrivial logical relations.
A \emph{$\Gamma$-formula} $\phi$ is a conjunction of constraints using  only
logical relations from $\Gamma$.
A
$\Gamma$-formula $\phi$ is satisfied by an assignment $m:\var(\phi)\to\{0,1\}$
if $m$ satisfies all constraints in $\phi$.

Throughout the text we refer to different types of Boolean relations following
Schaefer's terminology, see \cite{Schaefer78,creignou1997generating}. 
We say that a constraint language is \emph{Schaefer} if every relation in $\Gamma$ is either $\horn$, $\dualhorn$, $\bijunctive$, or $\affine$. 

\medskip
\centerline{
\fbox{
\begin{tabular}{ll}
 \multicolumn{2}{l}{$\pEnumSat{\Gamma}$} \\
  {Instance}: & $\phi$ a $\Gamma$-formula\\ 
  {Output}: & All satisfying assignments of $\phi$.\\
\end{tabular}
}}
\medskip

The following theorem gives the complexity of this problem according to $\Gamma$.
\begin{theorem}\label{thm:dichotomy}
 Let $\Gamma$ be a finite constraint language. 
 If $\Gamma$ is Schaefer, then $\pEnumSat{\Gamma}$ is in $\delayP$, 
 otherwise it is $\delayP_p^{\NP}$-complete via
 $D$-reductions and $\incP^{\NP}$-complete via $I$-reductions.
\end{theorem}
\begin{proof}
The polynomial cases were studied in \cite{creignou1997generating}.
Let us now consider the case where $\Gamma$ is not Schaefer.
Membership of $\pEnumSat{\Gamma}$ in $\delayP_p^{\NP}$ is clear. For the hardness, let us introduce 
$\T$ and $\F$ as the two unary constant relations $\T = \{1\}$ and $\F =
\{0\}$. 
According to Schaefer's dichotomy theorem \cite{Schaefer78}, deciding whether a 
$\Gamma\cup\{\F, \T\}$-formula is satisfiable is $\NP$-complete.
Since this problem is self-reducible, according to Theorem~\ref{theorem:selfreducomplete},
$\pEnumSat{\Gamma\cup\{\F, \T\}}$ is $\delayP_p^{\NP}$-complete via $D$-reductions. 
Let $\Imp$ denote the binary implication relation $\Imp(x,y)\equiv (x\longrightarrow y)$.
Given a $\Gamma\cup\{\F, \T\}$-formula $\varphi$, one can construct a $\Gamma\cup\{\Imp\}$ 
formula $\varphi'$ as follows. Let $V$ denote the set of variables in $\varphi$, $V_0$ the set of variables $ x$ on which lies the constraint $F(x)$, and $V_1$ the set of variables $x$  on which lies the constraint $T(x)$. Let $f$ and $t$ be two fresh variables. Given a set of variables $W$ and a variable $z$, $\varphi[W/z]$ denotes the formula obtained from $\varphi$ in replacing each occurrence of a variable from $W$ by the variable $z$. Then we set
$\varphi'=\varphi[V_0/f]\land \varphi[V_1/t]\land\bigwedge _{x\in V} \Imp(x, t)\land\bigwedge _{x\in V} \Imp(f,x)$. Clearly models of $\varphi$ coincide  with models of $\varphi'$ in which $f$ is set to false and $t$ to true, and $\varphi'$ might have two additional models, the all-zero one and the all-one one. In \cite{creignou1997generating} it is essentially proven that if $\Gamma$ is not Schaefer then the implication relation $\Imp$ can be expressed by a $\Gamma$-formula. Therefore, in this case the construction described above shows that  $\pEnumSat{\Gamma\cup\{\F, \T\}}\leq_e\pEnumSat{\Gamma}$, thus proving the  $\delayP_p^{\NP}$-completeness 
via $D$-reductions and $\incP^{\NP}$-completeness via $I$-reductions of the latter problem, 
according to Proposition \ref{prop:reduction_E_compared_to_D}. 
%\todo[inline]{NC: the proof uses the new definition of an e-reduction}
\end{proof}

To come back to the above discussion,  we point  out that there exist constraint languages $\Gamma$ such that the decision problem $\SAT{\Gamma}$ is in $\Pcomp$, 
while the enumeration problem  $\pEnumSat{\Gamma}$ is $\delayP_p^{\NP}$-complete, namely 0-valid or 1-valid constraint languages that are not Schaefer.

\subsection{Enumeration of Minimal Models}

Let us now turn to the complexity of enumerating minimal models of a Boolean formula. Such an enumeration   is ubiquitous in practical settings, ranging from verification, to  databases  and to knowledge representation among many others.
One can consider subset-minimality as well as cardinal-minimality. Therefore we consider the two following enumeration problems: $\circumscription^e$ denotes the problem of enumerating all
 subset-minimal models of a boolean formula and $\cardminsat^e$ denotes the problem of enumerating all
cardinality-minimal models of a boolean formula. We pinpoint the complexity of these two enumeration problems.

Observe that contrary to the satisfiability problems we discussed above, enumeration of minimal models is not obviously self-reducible. In particular, 
the problem
$\circumscription^e$ is very unlikely to be
self-reducible: The problem of deciding if a partial
truth assignment can be extended to a subset minimal model
is $\SigmaP{2}$-complete \cite{EG1995}, while deciding the
existence of a minimal model is clearly $\NP$-complete.
Thus $\circumscription^e$ is not self-reducible unless the polynomial
hierarchy collapses to the first level. 

\begin{theorem}\label{prop:circumscription}
  $\circumscription^e$ is $\incP^{\NP}$-complete via $I$-reductions. 
\end{theorem}
\begin{proof}
 Hardness follows immediately from Theorem~\ref{theorem:selfreducomplete},
 since deciding the existence of a (minimal) model of a Boolean formula is
 $\NP$-complete.

 To obtain membership, we show that $\circumscription^e$ can be even enumerated
 in $\delayP^\NP$. 
 Consider a boolean formula $\phi(x_1,\ldots,x_n)$, and assume that we want to
 enumerate the minimal models of $\phi$. We start by copying $\phi$ to the oracle
 registers. 
 A very first minimal model $x'$ can be achieved by a greedy algorithm using an $\NP$
 oracle. Next we extend $\phi$ in the oracle registers to
 $\phi'=\phi\wedge((x_1<x_1')\vee\ldots\vee(x_n<x_n'))$ and again get a minimal model
 $x''$ for $\phi'$ using a greedy algorithm with an $\NP$-oracle. This is also a minimal
 model for $\phi$; repeatedly extending $\phi$ and then computing a minimal model via
 a greedy algorithm achieves the membership $\circumscription^e\in\delayP^\NP$
 and thus also $\circumscription^e\in\incP^\NP$.
\end{proof}

What makes the  completeness result obtained for  $\circumscription^e$ surprising is
the discrepancy from the behaviour of the counting variant of the problem:
The counting variant of $\circumscription^e$ is a prototypical
$\#\cdot\coNP$-complete problem \cite{DHK2005}, and thus of the
same hardness as the counting variant of $\piSATe{1}$. However,
for enumeration we have that $\circumscription^e$ shows the same 
complexity as $\sigmaSATe{1}$, which is considered to be lower
than that of $\piSATe{1}$.

\begin{theorem}\label{theo:cardminsat}
  $\cardminsat^e$ is  $\delayP_p^{\NP}$-complete via
 $D$-reductions and $\incP^{\NP}$-complete via
 $I$-reductions. 
\end{theorem}
\begin{proof}
Hardness follows again immediately from Theorem~\ref{theorem:selfreducomplete}.

Membership of $\cardminsat^e$ in $\delayP_p^{\NP}$ (and a fortiori in  $\incP^\NP$)
follows from the fact that polynomial time with access to an $\NP$-oracle gives us
the possibility to compute first the minimal cardinality of models. Having this
information we can enumerate all cardinality-minimal models by the standard binary
search tree with an $\NP$-oracle that enumerates them in lexicographical order.
\end{proof}

% \todo[inline]{Should we say a word about the problem $\cardminsat$?}
% 

\subsection{Model-based Diagnosis}
\label{sect:diagnosis}

The problem of model-based diagnosis \cite{Reiter87} is ubiquitous in practical settings. A usual formulation of this problem is as follows: given a system description composed of some components, and an observation inconsistent with the system description, the goal is to identify a cardinality-minimal set of components which, if declared faulty and thus put out of the system description, 
restores the consistency between the model and the observation. In recent years the practical enumeration of minimum-size diagnoses has gained interest through translations to MaxSAT instances and the use of efficient MaxSAT solvers (see e.g. \cite{Marques-SilvaJI15}). 

Formally a system description can be represented by a  consistent set of propositional formul\ae, each encoding the normal behavior of a component. The observation is represented by an additional  formula that is inconsistent with the former set of formul\ae. Thus we get the following enumeration problem.

\medskip
\fbox{
\begin{tabular}{ll}
 \multicolumn{2}{l}{$\mbd$} \\
  {Instance}: & ${\cal B}=\{\phi_1, \ldots, \phi_n\}$ a consistent set of formul\ae,\\& $\mu$ a satisfiable formula inconsistent with $\calB$\\ 
  {Output}: & All cardinality-maximal subsets ${\cal B'}\subseteq {\cal B}$ such that \\
 & $\displaystyle\bigwedge_{\phi\in {\cal B'}}\phi\land \mu$ is satisfiable.\\
\end{tabular}
}
\medskip

Observe that for this enumeration problem, the corresponding decision problem $\pExist{R}$ is trivial because there always is a solution $\calB'$. Nevertheless we prove hardness by an $e$-reduction from a previously identified hard problem. 
So contrary to all results above, we obtain here a hardness result for an enumeration problem which does not rely on hardness of any related decision problem.

The following theorem settles the complexity of $\mbd$ in our hierarchy.

\begin{theorem}\label{prop:diagnosis}
 $\mbd$ is  $\delayP_p^{\NP}$-complete via
 $D$-reductions and $\incP^{\NP}$-complete via
 $I$-reductions.
\end{theorem}
\begin{proof}
Membership of $\mbd$ in $\delayP_p^{\NP}$ (and a fortiori in  $\incP^\NP$) follows from the fact that polynomial time with access to an $\NP$-oracle gives us the possibility to compute first the maximal cardinality of subsets of ${\cal B}$ that are consistent with $\mu$.  Having this information we can enumerate all such cardinality-maximal subsets by the standard binary search tree with an $\NP$-oracle  that enumerates them in lexicographical order.

For hardness we show that $\cardminsat^e\le_e\mbd$, thus proving that $\mbd$ is 
$\delayP_p^{\NP}$-hard via
$D$-reductions and $\incP^{\NP}$-hard via
$I$-reductions, according to Proposition~\ref{prop:reduction_E_compared_to_D} and
Proposition~\ref{prop:circumscription}. 
Let $\phi$ be a formula given as an instance of $\cardminsat^e$, and assume
$\var(\phi)=\{x_1, \ldots, x_n\}$. Without loss of generality suppose in addition
that $\phi$ is not 0-valid, i.e., is not satisfied by the all-zero assignment.
We define a corresponding instance $\sigma(\phi)$ of $\mbd$ as follows:
Let ${\cal B}$ consist of the $n$ formul\ae\  $\phi'_i=(\phi\lor x_0)\land (\neg x_i\lor x_0)$ for $i=1,\ldots, n$,  and let $\mu=\neg x_0$, where $x_0$ is a fresh variable. Observe that on the one hand the set ${\cal B}$ is consistent (set $x_0$ to true), and on the other hand it is inconsistent with $\mu$ (since by assumption $\phi$ is not 0-valid). It is easy to see that there is a bijection $\pi$ mapping the cardinality-minimal models $x'$ of $\phi$ to cardinality-maximal subsets of ${\cal B}$ that are consistent with $\mu$. Thus we can define a relation $\tau$ with $\tau(\sigma(\phi),x',-)=\pi(x')$,
showing that  $\cardminsat^e\le_e\mbd$.
\end{proof}

\subsection{Abduction}

We now provide an example of an enumeration problem hard in a higher level of the hierarchy, namely abduction. Abduction is a well-studied form a non-monotonic reasoning. Given a certain consistent knowledge, abductive reasoning is used to generate explanations for observed manifestations. We consider here propositional abduction, in which the knowledge base as well as the manifestations are represented by propositional formul\ae\  \cite{EG1995}. 

The enumeration problem can be formalized as follows:

\medskip
\fbox{
\begin{tabular}{ll}
 \multicolumn{2}{l}{$\abd$} \\
  {Instance}: & $\Gamma$ a set of formul\ae\ (knowledge),\\& $H$ a set of literals (hypotheses),\\& $q$ a variable that does not occur in $H$ (manifestation).\\ 
  {Output}: & All sets $E\subseteq H$ such that $\Gamma\land E$ is satisfiable and\\&
  $\Gamma\land E\models q$ (equivalently, $\Gamma\land E\land\neg q$ is unsatisfiable). \\
\end{tabular}
}
\medskip

The decision problem $\abddec$ corresponding to this variant of abduction is known to be $\SigmaP{2}$-complete, see, e.g., \cite{CZ06}. Based on this we obtain hardness of the  enumeration problem. Though this is maybe not very surprising, we mention this problem here as a first example of a problem whose enumeration is even harder than those we studied before, at least under the assumption that the polynomial hierarchy does not collapse.

\begin{theorem}
$\abd$ is  $\delayP_p^{\SigmaP{2}}$-complete via
$D$-reductions and $\incP^{\SigmaP{2}}$-complete via
$I$-reductions.
\end{theorem}
\begin{proof}
As we have mentioned above, the associated decision problem is $\SigmaP{2}$-complete. Also, note that the problem is self-reducible in the sense of Definition~\ref{def:self-reducible}.
Given an instance $(\Gamma,H,q)$ of $\abd$ and a partial solution $E\subseteq H$, $E$ can be extended to an explanation iff $(\Gamma\land E,H,q)\in\abddec$.
Therefore, according to the last statement of Theorem~\ref{theorem:selfreducomplete}, the theorem follows.
\end{proof}

% \todo[inline]{May be one can find an enumeration problem from database, repairs of databases?, for which we can find the complexity}

\subsection{Database Repairs}

Finally, we provide an example of an enumeration problem from the database domain, more specifically the problem of enumerating all ``repairs'' of an inconsistent database. 
In this problem, we are given a {\em database instance\/} $D$ 
(i.e, simply a finite structure given by a domain and a collection of relations; for our purposes, it is most convenient to represent $D$ as a set of
ground atoms) and a set $C$ of {\em integrity constraints\/}  (i.e., in general, a 
set of formul\ae\  from some fragment of first-order logic).
Now suppose that $D$ is {\em inconsistent\/}, i.e., $ D \not\models C$ holds. 
Then we are interested in enumerating the possible repairs. A {\em repair\/}
is a database instance $I$ which satisfies $C$ and which \emph{differs minimally} from $D$.
The term ``differs minimally'' can be defined in several ways.
The most common approach is the one introduced by 
Arenas et al.~\cite{DBLP:conf/pods/ArenasBC99}: here, 
repairs are obtained from the original database by inserting or deleting 
ground atoms to or from $D$, respectively. Then ``minimal difference'' means 
that the symmetric set difference $\Delta (D,I)$ is minimal w.r.t.\ subset inclusion. 
More formally, let 
$\Delta (D,I) = (D \setminus I) \cup (I \setminus D)$. 
Then  $I$ is a repair of  $D$ w.r.t.\ $C$ if $I \models C$ 
(i.e.,  $I$ is consistent w.r.t.\ $C$)
and 
there does not exist an instance $I'$ that satisfies $C$ and 
$\Delta (D,I') \subsetneq \Delta (D,I)$.

Many different fragments of first-order logic have been studied for the definition of 
the integrity constraints $C$. We only look at one such fragment here, namely 
{\em equality generating dependencies\/} (EGDs, for short).
EGDs are a classical formalism for defining database constraints \cite{AHV}.
They generalize functional dependencies (FDs), which in turn generalize 
key dependencies (KDs). Formally, they are defined as formul\ae\ of the form 
\[
    \forall \vec{x}\big(\phi(\vec{x}) \rightarrow x_i = x_j\big),
\]
where $\phi$ is a conjunction of atoms over the predicate symbols occurring in $D$, 
$\vec{x}$ denotes the variables occurring in these atoms, and 
$x_i, x_j$ are two variables from $\vec{x}$. The meaning of such an EGD is that, 
for every variable binding 
$\lambda$ on the variables $\vec{x}$, such that applying $\lambda$ to the atoms 
in $\phi(\vec{x})$ sends all atoms of $\phi(\vec{x})$ into the database instance $D$, it must
be the case that $\lambda(x_i) = \lambda(x_j)$ holds. 
Following the usual notational convention in the database literature, we will 
omit the quantifiers and denote EGDs of the above form simply 
as $\phi(\vec{x}) \rightarrow x_i = x_j$ with the 
understanding that all variables occurring in an EGD are universally quantified.

The enumeration problem considered here is defined as follows:

\medskip
\fbox{
\begin{tabular}{lll}
 \multicolumn{2}{l}{$\repair$} \\
  {Instance}: & Database instance $D$, set $C$ of EGDs\\ 
  {Output}: & $D$ \hskip115pt if $D \models C$ and\\
            & All repairs $I$ of $D$ w.r.t.\ $C$ \hskip10pt if $D \not\models C$.\\
\end{tabular}
}
\medskip

As in Section \ref{sect:diagnosis}, we observe that for this enumeration problem, the corresponding decision problem $\pExist{R}$ is trivial because for every inconsistent database instance, the existence of a repair is guaranteed.
Moreover, it is worth mentioning that a different variant of repair enumeration has been recently studied in 
\cite{DBLP:conf/pods/LivshitsK17}, where a different type of integrity constraints was considered and prereferences were taken into account.

Several decision problems have been studied in the context of repairs. 
The most intensively studied one is the {\sc Consistent Query Answering} problem, i.e., 
given a database instance $D$, a set $C$ of integrity constraints and a boolean query $Q$, %(usually, restricted to conjunctive queries), 
is the query $Q$ true in every repair $I$ of $D$, w.r.t.\ $C$. 
For EGDs as constraint language and conjunctive queries as query language, 
this problem has been 
shown to be $\PiP{2}$-complete \cite{DBLP:conf/icdt/ArmingPS16}. 
Another decision problem
% , which comes closer to our enumeration problem $\repair$ 
is the {\sc Repair Checking} problem, i.e., 
given database instances $D$ and $I$, and a set $C$ of integrity constraints, 
is $I$ a repair of $D$ w.r.t.\ $C$. 
In other words, this is the {\sc Check} problem corresponding to the 
enumeration problem $\repair$.
For EGDs, this decision problem is
{\sf DP}-complete \cite{DBLP:conf/icdt/ArmingPS16}. 
Definitely a bit more surprising than in case of $\abd$, we get 
$\incP^{\SigmaP{2}}$-completeness also for the $\repair$ problem. 
In particular, the hardness-proof is more involved than for 
the previously studied problems.

\begin{theorem}
$\repair$ is  $\incP^{\SigmaP{2}}$-complete via $I$-reductions.
\end{theorem}
\begin{proof}
For the {\it membership} proof, it is important to note that 
repairs of $D$ w.r.t.\ a set of EGDs can only be obtained by deletions, i.e., 
if a constraint is violated, it does not help to add further atoms.
In other words, EGDs have a monotonic behaviour in the sense that if a database instance $I$ 
violates an EGD, then also every superset of $I$ does. 
This allows us to redefine $\repair$ as $\pEnum{R}$, where the relation $R$ consists of pairs $((D,C),V)$ 
such that
\begin{itemize}
\item $D$ is a database instance containing $N\geq0$ atoms, which are numbered
as $A_1, \dots, A_N$,
\item $C$ is a set of EGDs and
\item $V$ is a word of length $N$ over $\{0,1\}$, where $V(i) = 0$ means that we delete $A_i$ from $D$ and
$V(i) = 1$ means that we retain $A_i$ in $D$.
\end{itemize}
We observe that $\pExtSol{R}\in\SigmaP{2}$: Given some word $V_1\in\{0,1\}^{N-k}$, we need 
a non-deterministic guess for the instantiation of some $V_2\in\{0,1\}^k$ and then we 
need to call a {\sf DP}-oracle (or, equivalently, two 
{\sf NP}-oracles)
to check if the word $V_1V_2$
represents a repair. Thus $\pEnum{R}\in\delayP^{\SigmaP{2}}_p$ by Proposition~\ref{prop:new}, meaning
that $\repair\in\incP^{\SigmaP{2}}$.

\smallskip

For the {\it hardness} proof, we first note 
that the $\incP^{\SigmaP{2}}$-completeness
of $\piSATe{1}$ via $I$-reductions from Corollary~\ref{cor:complexity_sat}
can be strengthened in the sense that we may restrict the 
instances to 3-DNF (i.e., the matrix of the formul\ae\ under investigation 
is in DNF such that each implicant consists of 3 literals.
Then our $\incP^{\SigmaP{2}}$-completeness proof is by an $e$-reduction from 
$\piSATe{1}$, where we assume w.l.o.g., that the instances of $\piSATe{1}$
are in 3-DNF.

Consider an arbitrary instance of $\piSATe{1}$, i.e., 
a quantified boolean formula $\psi(\vec{x}) = \forall \vec{y} \phi (\vec{x} ,\vec{y})$.
Let $\vec{x} = x_1, \dots, x_k$ and 
$\vec{y} = y_1, \dots, y_\ell$.
By the restriction to 3-DNF, $\phi (\vec{x} ,\vec{y})$
is of the form $\phi (\vec{x} ,\vec{y}) = \bigvee_{i=1}^m 
(l_{i1} \wedge l_{i2} \wedge l_{i3})$, where $m$ denotes the number of implicants in
$\phi$ and $l_{ij}$ with $j \in \{1,2,3\}$ is the $j$-th literal in the $i$-th implicant, i.e., each $l_{ij}$ is of the form $x_\alpha$, $\neg x_\alpha$,
$y_\beta$, or $\neg y_\beta$ with $\alpha \in \{1, \dots, k\}$ and 
$\beta  \in \{1, \dots, \ell\}$.

From this we define an instance $\sigma(\psi(\vec{x}))=(D,C)$ of $\repair$.
The relation symbols used in this instance are 
$\{a,q\} \cup \{p_1, \dots, p_k\} \cup \{b_0,b_1\}$
It is convenient to 
define the following subformul\ae: 

\newcommand*{\ol}[1]{\overline{#1}}

\setcounter{equation}{0}

\begin{align*}
    \chi={}& b_0(v_0) \wedge b_1(v_1) \wedge q(v_0,v_1) 
                               \wedge q(v_1,v_0) \wedge \mbox{} \\
                     & a(v_0,v_0,v_0) \wedge a(v_0,v_0,v_1)\wedge a(v_0,v_1,v_0)
                                \wedge a(v_0,v_1,v_1) \wedge  \mbox{} \\       
                     & a(v_1,v_0,v_0) \wedge a(v_1,v_0,v_1)\wedge a(v_1,v_1,v_0) \\
    \pi={}& \bigwedge_{i=1}^k p_i(w_i,w'_i)\\
    \psi^*={}& \bigwedge_{i=1}^m a(l^*_{i1},l^*_{i2},l^*_{i3})
\end{align*}

\noindent
where we define $l^*_{ij} = z$ if $l_{ij}$ is the positive literal 
$z$ and $l_{ij}^*=z'$ if $l_{ij}$ is the negative literal 
$\neg z$.
For example,
$[(x \wedge \neg y \wedge z) \vee (\neg x \wedge y \wedge z)]^* = 
a(x,y'z) \wedge a(x',y,z)$. 

\smallskip

\noindent
Then the database $D$ of our instance of $\repair$ looks as follows:
\begin{align*}
    D={}& P \cup Q \cup A \cup E \mbox{ with } \\
    P={} & \{p_1(0,1), p_1(1,0),p_2(0,1), p_2(1,0), \dots, p_k(0,1), p_k(1,0)\}\\
    Q={} & \{q(0,1), q(1,0)\}\\
    A={} & \{a(0,0,0), a(0,0,1), a(0,1,0),a(0,1,1), a(1,0,0), a(1,0,1), a(1,1,0)\} \\
    E={} & \{b_0(0), b_1(1)\}\\[1.1ex]
\end{align*}
\noindent
Finally, we define the set $C$ of integrity constraints: 
\begin{align}
    C={}& \bigcup_{1\leq i\leq k} \{ \chi \wedge \pi \wedge p_i(x,x') \wedge p_i(x',x) 
               \rightarrow  x = x' \} \tag{C1}\\
        & \cup \{ \chi \wedge \pi \wedge \bigwedge_{i=1}^k p_i(x_i,x'_i)  \wedge 
                  \bigwedge_{j=1}^\ell q(y_j,y'_j) \wedge \psi^* 
                   \rightarrow x_1 = x'_1 \}\tag{C2}
\end{align}
The intuition of this reduction is as follows: 
all of the EGDs in line (C1) are clearly violated, i.e.: $\chi$ can be sent into the 
database $D$ by instantiating $v_0$ to 0 and $v_1$ to 1; 
likewise $\pi$ can be sent into the database $D$ by instantiating each pair $(w_i,w'_i)$
either to $(0,1)$ or to $(1,0)$. Finally, 
each of the atoms $p_i(x,x')$ and $p_i(x',x)$ can be sent into $D$ by instantiating  $(x,x')$ 
to either $(0,1)$ or $(1,0)$. 
In either case, the left-hand side of the EGD is satisfied but the right-hand side is not.
Hence, there are three possibilities to eliminate this EGD-violation: 
\begin{description}
\item[(D1)] we can delete one of the atoms in $Q \cup A \cup E$, 

\item[(D2)] we can delete both atoms $p_i(0,1)$ and $p_i(1,0)$ for some 
$i \in \{1, \dots, k\}$, or

\item[(D3)] we can delete one of the atoms $p_i(0,1), p_i(1,0)$ for every $i \in \{1,\dots, k\}$.

\end{description}

%We now analyse the nature of the repairs resulting from each of these 3 possibilities: 
%(1) If we delete one of the atoms in $Q \cup A \cup E$, then there can be 
%no violation of the EGD in line (2). Hence, every instance 
%$I = D \setminus \{\mathit{at}\}$ with $\mathit{at} \in Q \cup A \cup E$ is a repair, since
%$I$ clearly satisfies the  minimality condition on repairs. 
%(2) If we delete both atoms $p_i(0,1)$ and $p_i(1,0)$ for some 
%$i \in \{1, \dots, k\}$, then again the EGD in line (2) is satisfied
%and every instance 
%$I = D \setminus \{p_i(0,1), p_i(1,0)\}$ with $i \in \{1, \dots, k\}$ is a repair
%(again the minimality condition on repairs is easy to verify).
%It remains to consider case (3):
%

Both (D1) as well as (D2) lead to a repair, as the subformul\ae\ $\chi$ respectively $\pi$ are
always violated, so we need to consider (D3). The effect
of retaining $p_i(0,1)$ with $i \in \{1, \dots, k\}$ is that
the variables
$(x_i,x'_i)$ in line (C2) are instantiated to (0,1); likewise, 
the effect 
of retaining $p_i(1,0)$ is that the 
the variables
$(x_i,x'_i)$ in line (C2) are instantiated to (1,0). Hence, there is a one-to-one 
correspondence between the choices of retained ground atom $p_i(.,.)$ for every 
$i \in \{1, \dots, k\}$ and truth assignments to the $x_i$-variables in $\psi$, namely
retaining $p_i(0,1)$ (resp.\ $p_i(1,0)$) corresponds to 
setting the propositional variable $x_i$ in $\psi$ to 
false (resp.\ true) and to set its dual $\neg x_i$, which is encoded by $x'_i$ according to our definition of $\psi^*$, to true (resp.\ false).  For a repair $I$ that only uses deletions of the form (D3),
denote by $\rho(I)$ the corresponding truth assignment on the propositional variables $\{x_1, \dots, x_k\}$.
We then define the relation $\tau$ for the $e$-reduction as follows:
\begin{equation*}
\tau(\sigma(\psi(\vec{x})), I,-)=\left\{
     \begin{array}{ll}
       \emptyset & : I\text{ is the result of a deletion of the form (D1) or (D2)},\\
       \rho(I) & :\text{otherwise}.
     \end{array}
   \right.
\end{equation*}

To show that $\tau$ and $\sigma$ indeed define a valid $e$-reduction from $\piSATe{1}$,
first note that $\tau(\sigma(\psi(x)), I,-)$ maps to the empty set at most $k+11$ times. So it suffices to
show that
the repairs we can get if we apply a deletion of the form (D3) are indeed models of $\psi$.
We distinguish two cases: 
First suppose that the EGD in line (C2) is violated by the instance resulting from the deletion of exactly one of $p_i(0,1), p_i(1,0)$ for every $i \in \{1,\dots, k\}$. Then we 
have to delete at least one more atom from $D$. This means that we either delete 
one atom from $Q \cup A \cup E$ to falsify $\chi$ (and possibly also further atoms in 
$q(y_i,y'_i) \wedge \psi^*$)  or we delete 
the remaining atom with leading symbol $p_i$ to falsify 
atom $p_i(x_i,x'_i)$ on the left-hand side of the EGD. In either case, the resulting 
database instance is not a repair due to the minimality condition. 

It remains to consider the case that the EGD in line (C2) is satisfied 
by the instance $I$ resulting from the deletion of exactly one of $p_i(0,1), p_i(1,0)$ for every $i \in \{1,\dots, k\}$. Clearly, the equality on the right-hand side is always violated, no matter how we instantiate $(x_1,x'_1)$, since the only options are $(0,1)$ and $(1,0)$.
This means that the only way to satisfy the EGD in line (C2) is that there is no way to send the left-hand side into the database instance $I$. This means that no matter how we
instantiate $(y_j,y'_j)$ -- i.e., either to $(0,1)$ or to $(1,0)$ --
there is always one atom in $\psi^*$ that is instantiated to an $a$-atom outside $I$, 
namely $a(1,1,1)$. By our encoding of $\psi$ by $\psi^*$ this has the following 
meaning for the quantified boolean formula $\psi$: no matter how we extend
the truth assignment $\rho(I)$ to an assignment on 
$\{y_1, \dots, y_\ell\}$ (corresponding to the possible ways of instantiating $(y_j,y'_j)$
for every $j \in \{1, \dots, \ell\}$),
there is  at least one implicant in $\psi$ for which all three literals evaluate to 
true. In other words, there is a one-to-one correspondence between repairs 
obtained by deleting exactly one of $p_i(0,1), p_i(1,0)$ for every $i \in \{1,\dots, k\}$
and models of $\psi$.
%
%our reduction from 
%$\piSATe{1}$ to $\repair$ is in fact a $\leq_e$-reduction (and, therefore, 
%also a $\leq_I$ reduction), where the instance of 
%$\repair$ has $2 k +11$ solutions without corresponding solution of $\piSATe{1}$
%(easily recognizable by the fact that either one of the 11 atoms of 
%$Q \cup A \cup E$ is missing in the repair or both atoms 
%$p_i(0,1), p_i(1,0)$ for every $i \in \{1,\dots, k\}$ are missing). 
%For the remaining solutions of $\repair$, there is a one-to-one correspondence with the solutions of $\piSATe{1}$.
\end{proof}

% 
% \section{OutputP}
% 
% \todo[inline]{Redefine the declarative style reduction, there is one in the thesis of Mary.}
% 
% \todo[inline]{Give an example for the new reduction}
% 
% \todo[inline]{Discuss the Procedural-style Reductions in terms of OutputP:
% I.e., introduce the exhaustive reductions}
% 
% \todo[inline]{Demonstrate the usefulness of any of these reductions on
% interesting, natural problems}
% 
% \todo[inline]{Find/Discuss interesting problems in OutputP that are not
% known to be in IncP.}
% 
% \todo[inline]{Find/Discuss interesting problems in OutputP$^\NP$}
% 

\section{Conclusion}
\label{sect:conclusion}

We introduced a hierarchy of enumeration complexity classes, extending the well-known tractable enumeration classes $\delayP$ and $\incP$, just as the $\DeltaP{k}$-classes of the polynomial-time hierarchy extend the class $\ptime$. We show that under reasonable complexity assumptions these hierarchies are strict. We introduced a declarative type of reduction among enumeration problems which is transitive and under which the classes in our hierarchies are closed.
Moreover, they turn out to be special cases of the procedural-style $D$- and $I$-reductions that we defined afterwards. These procedural style reductions allow one to get completeness results for the classes of our hierarchy. The $e$-reductions are conceptually simpler and easier to obtain; so they will be used to obtain further hardness results for enumeration problems. 
In this way we obtain completeness results for diverse problems such as
generalized satisfiability in the Schaefer framework, 
circumscription, 
model-based diagnosis, 
abduction, and repairs.

Prior to our work,  lower bounds for enumeration problems were only of the form ``$\pEnum{R}$ is not in $\delayP$ (or $\incP$) unless $\ptime=\NP$''.
We have provided here a framework which allows one to pinpoint the complexity of such problems in a better way in terms of completeness in our enumeration hierarchy. 

As far as future work is concerned, of course, we want to close the gap left in Figure~\ref{fig:hierarchy}, 
i.e., clarifying the precise relationship between the enumeration complexity classes 
$\delayP^{\SigmaP{k}}$ and $\delayP^{\DeltaP{k+1}}$. 
For a wider field of research activities in the area of hard enumeration problems, recall that our 
hierarchies of enumeration complexity classes build upon the tractable classes $\delayP$ and $\incP$. Note that in 
\cite{JPY1988}, yet another notion of tractable enumeration is proposed, namely the class $\outputP$ (also referred to as $\totalP$ 
in \cite{Strozecki2010}). A problem is in $\outputP$ if the time needed to output all solutions to a given instance 
is bounded by a polynomial in the combined size of the input {\em plus\/} the output. It is easy to show that 
this class  is closed w.r.t.\ $e$-reductions and only minor modifications of the $D$- and $I$-reductions are needed to close 
$\outputP$  also under these reductions. However, while we have seen quite a close relationship between $\delayP$ and $\incP$ 
in this work, the class $\outputP$ seems to behave very differently from the others. For instance, $\outputP$ seems 
to be incomparable (under common complexity-theoretic assumptions)
with {\em all\/} higher classes from our hierarchies built upon $\delayP$ or $\incP$. Sure, we may also define a hierarchy of classes 
by allowing, for instance, a $\SigmaP{k}$ oracle for some $k \geq 1$ on top of $\outputP$. However, it is unclear what complete problems for any of the resulting classes would look like. A major research effort is required to even get a basic understanding of such classes. 

Finally, a natural next research target is to put the machinery developed here to work. 
In Section \ref{sect:completeness}, we have already proved a few completeness 
results for our new complexity classes. Many more intractable enumeration problems (especially in the AI and database domains) wait
for a precise complexity classification. We believe that such work -- complementing the search for restrictions to make intractable
enumeration problems tractable, cf.\ \cite{KPS2016,DBLP:conf/pods/LivshitsK17} -- is needed to get a better understanding of the true sources of complexity 
of hard enumeration problems.

\section*{Acknowledgments}

We  would like to thank Phokion Kolaitis for his encouragement on this work during his visit to Marseilles 
and for suggesting to consider  database repairs  as complete problems for our complexity classes.

This work was supported by the
Vienna Science and Technology Fund (WWTF) through project ICT12-015,
the Austrian Science Fund (FWF): P25207-N23, P25518-N23, I836-N23, W1255-N23
and the French Agence Nationale de la Recherche,
AGGREG project reference ANR-14-CE25-0017. 

\bibliography{references}

\end{document}